\documentclass[a4paper,11pt]{article}
\usepackage[top=1 in,bottom=1 in,left=1.0 in,right=1.0 in]{geometry}
\usepackage{float}
\usepackage{cite}
\usepackage{graphicx}
\usepackage{subfigure}
\usepackage{amssymb,amsmath,amsthm,amscd}
\usepackage{mathrsfs}
\usepackage{tipa}
\usepackage{epic}
\usepackage[colorlinks,linkcolor=blue,citecolor=blue]{hyperref}

\newcommand{\PT}{$\mathcal{PT}$}

\newtheorem{theorem}{Theorem}
\newtheorem{corollary}{Corollary}%[section]
\newtheorem{remark}{Remark}
\title{Solutions of local and nonlocal equations reduced from the AKNS hierarchy}

\author{Kui Chen,~ Xiao Deng,  ~ Senyue Lou,~ Da-jun Zhang\footnote{Corresponding author. Email: djzhang@staff.shu.edu.cn}\\
{\small $^{1}$Department of Mathematics, Shanghai University, Shanghai 200444, P.R. China}\\
{\small $^{2}$Faculty of Science, Ningbo University, Ningbo 315211, P.R. China}}

\begin{document}

\maketitle

\begin{abstract}

  In the paper possible local and nonlocal reductions of the Ablowitz-Kaup-Newell-Suger (AKNS) hierarchy are collected,
  including the Korteweg-de Vries (KdV) hierarchy, modified KdV hierarchy  and their nonlocal versions,
  nonlinear Schr\"{o}dinger hierarchy  and their nonlocal versions,
  sine-Gordon equation in nonpotential form and its nonlocal forms.
  A reduction technique for solutions is employed,
  by which exact solutions in double Wronskian form are obtained for these reduced equations from
  those double Wronskian solutions of the AKNS hierarchy.
  As examples of dynamics we illustrate new interaction of two-soliton solutions of the reverse-$t$ nonlinear Schr\"{o}dinger equation.
  Although as a single soliton it is always stationary, two solitons travel along completely symmetric trajectories in $\{x,t\}$ plane
  and their amplitudes are affected by phase parameters.
  Asymptotic analysis is given as demonstration.
  The approach and relation described in this paper are systematic and general and can be used to other nonlocal equations.

\begin{description}
\item[PACS numbers:]
02.30.Ik, 02.30.Ks, 05.45.Yv
\item[Keywords:] reductions, nonlocal, AKNS hierarchy, double Wronskians, solutions
\end{description}
\end{abstract}

\section{Introduction}\label{sec-1}

Integrable systems play important roles in nonlinear mathematics and physics.
As a feature of integrability such systems posses their own exact multi-soliton solutions.
Developing techniques of finding exact solutions for nonlinear systems  is always a main topic in integrable systems.
Recently, an integrable system called nonlocal nonlinear Schr\"odinger (NLS) equation,
\begin{equation}
iq_{t}(x,t)+q_{xx}(x,t)\pm q^2(x,t)q^*(-x,t)=0,
\label{nx-NLS}
\end{equation}
was proposed by Ablowitz and Musslimani\cite{AM-PRL-2013}.
In this equation $q$ is a complex function of real arguments $(x,t)$,
$i$ is the imaginary unit and $*$ stands for complex conjugate.
Eq.\eqref{nx-NLS} is \PT-symmetric because the potential $V(x,t)=\pm q(x,t)q^*(-x,t)$
satisfies the \PT-symmetry condition $V(x,t) = V^*(-x,t)$.
To learn more about \PT-symmetry one can refer to review \cite{KYZ-RMP-2017}.
Eq.\eqref{nx-NLS} is also nonlocal as $x$ and $-x$ are involved simultaneously.
Similar type system, a nonlocal KdV model, is used to describe interaction of a two-place event \cite{lou-2017}.

Following \cite{AM-PRL-2013}, more integrable nonlocal systems were proposed \cite{SMM-PRE-2014,AM-PRE-2014,AM-Nonl-2016,ablowitz-study-2016,fokas-2016}.
Meanwhile, solving techniques are developed to deal with nonlocalness resulting from $-x$.
Among the main techniques are the Inverse Scattering Transform \cite{AM-PRL-2013,AM-Nonl-2016,zuo-JMAA-2017},
Darboux transformations  \cite{LiX-PRE-2015,Chen-JAMP-2015,HuaL-EPJP-2016,LiXM-JPSJ-2016,XU-AML-2017,zuo-jpsj-2017,ZhouZX-1,ZhouZX-2},
and the bilinear method\cite{Yan-AML-2015,XU-AML-2016,MaZ-AML-2016}.
Actually, it is hard to apply bilinearisation directly to Eq.\eqref{nx-NLS} itself, because it is difficult
to deal with the nonlocal element $q^*(-x,t)$ in bilinear approach.
Eq.\eqref{nx-NLS} is reduced from the second order Ablowitz-Kaup-Newell-Suger (AKNS)  coupled system
\begin{subequations}\label{brc}
\begin{align}
& iq_t=q_{xx}-2q^2r,\\
& ir_t=-r_{xx} +2r^2q
\end{align}
\end{subequations}
via the reduction of $r(x,t)=\mp q^*(-x,t)$.
The before-reduction system \eqref{brc} was already solved through bilinear  approach and solutions
were presented in terms of double Wronskians \cite{chen2008}.
Recently we developed a reduction technique
by imposing a constraint on the two basic vectors in double Wronskians
so that $q$ and $r$ obey the relation $r(x,t)=\mp q^*(-x,t)$
and therefore solutions to the nonlocal equation \eqref{nx-NLS} are obtained \cite{ChenZ-AML-2017}.
Similar approach is also available  to semidiscrete case \cite{deng-2017}.

There is a bilinear form for the whole isospectral AKNS hierarchy and solutions  are presented in terms of double
Wronskians \cite{liu1990}.
In this paper, starting from the double Wronskian solutions of the AKNS hierarchy,
we will consider suitable reductions on these solutions. As a result,
we will obtain solutions of the reduced equations, including the Korteweg-de Vries (KdV) hierarchy,
the NLS hierarchy, the nonlocal  NLS hierarchy with reverse space, reverse time and reverse space-time,
the mKdV hierarchy, the complex mKdV (cmKdV) hierarchy and their nonlocal counterparts, and the sine-Gordon equation (in nonpotential form) and its nonlocal version.
As an example of dynamics of obtained solutions,
we will illustrate  two-soliton interactions of reverse-$t$ nonlinear Schr\"{o}dinger equation.
Although as a single soliton it is always stationary, trajectories of two solitons are  completely symmetric in $\{x,t\}$ plane
and their amplitudes are affected by phase parameters.
Asymptotic analysis is given as demonstration.
We also illustrate solutions of reverse-$(x,t)$ sine-Gordon equation. They exhibit dynamics similar to the usual
mKdV equation.

The paper is organized as follows.
In Sec.\ref{sec-2} we recall local and nonlocal reductions of the AKNS hierarchy, and bilinear form and double-Wronskian solutions
of the AKNS hierarchy.
In Sec.\ref{section-reduction-of-solution} we derive the solutions for the systems reduced from the AKNS hierarchy
by considering suitable constraints on the elements in double Wronskians.
Then, in Sec.\ref{sec-4} we consider reductions of the first member in the negative AKNS
hierarchy to get sine-Gordon equation (in its nonpotential form).
In particular, we introduce a suitable integration operator to implement nonlocal reductions
and obtain nonlocal sine-Gordon equation.
Finally, Sec.\ref{sec-5} includes conclusions and discussions.

\section{The AKNS hierarchy and double Wronskian solutions}\label{sec-2}

In this section we list out some known results on local and nonlocal reductions\cite{AKNS-PRL-1973,Abl1974,ablowitz-study-2016}
of the isospectral AKNS hierarchy.
Double Wronskian solutions of the AKNS hierarchy \cite{liu1990,yin2008} will be also presented.

\subsection{The AKNS hierarchy and reductions}\label{sec-2-1}

It is well known that from the AKNS spectral problem\cite{AKNS-PRL-1973,Abl1974}
\begin{equation}\label{akns-spectral}
 \Phi_x =\left( \begin{array}{cc}
                                  \lambda & q \\
                                  r & -\lambda
                                \end{array}
                              \right)  \Phi,~~\Phi=\left(
                                \begin{array}{c}
                                  \phi_1 \\
                                  \phi_2
                                \end{array}
                              \right)
\end{equation}
with spectral parameter $\lambda$ and potentials $q=q(x,t)$ and $r=r(x,t)$,
one can derive the isospectral AKNS hierarchy
\begin{equation}\label{akns-hierarchy}
 \left(
   \begin{array}{c}
     q_{t_n} \\
     r_{t_n} \\
   \end{array}
 \right)      =K_n     = \left(
                      \begin{array}{c}
                        K_{1,n} \\
                        K_{2,n}
                      \end{array}
                    \right)
                    =L^n \left(
                      \begin{array}{c}
                        -q \\
                        r
                      \end{array}
                    \right),
                    ~~n=1,2,\cdots,
\end{equation}
where $L$ is the recursion operator denoted by
\begin{equation*}
%\label{akns-recusion-operator}
 L= \left(
            \begin{array}{cc}
              -\partial_x+2q \partial_x^{-1}r & 2q \partial_x^{-1}q \\
              -2r \partial_x^{-1}r & \partial_x- 2r \partial_x^{-1}q
            \end{array}
          \right),
\end{equation*}
in which $\partial_x=\frac{\partial}{\partial x}$, $\partial_x\partial^{-1}_x=\partial^{-1}_x\partial_x=1$.
In the following we list possible one-component equations reduced from \eqref{akns-hierarchy}  that will be considered in this paper.

\noindent
\textbf{KdV hierarchy}:
\begin{equation}
r_{t_{2l+1}}=K_{2l+1}^{KdV}=K_{2,2l+1}|_{\eqref{kdv-hie-reduction}},~~ l=0,1,\cdots,
\label{kdv-hie}
\end{equation}
reduced from those odd-indexed members of the AKNS hierarchy \eqref{akns-hierarchy} with  reduction
\begin{equation}\label{kdv-hie-reduction}
  q(x,t)= \pm 1,
\end{equation}
in which the representative is the KdV equation
\begin{equation}
r_{t_3}=K_{3}^{KdV}=r_{xxx} \mp 6rr_x.
\label{kdv-eq}
\end{equation}
Here, the sign $\pm$ does not make sense in equation as it is easy balanced by $r \to \pm r$.
However, we will keep that for later convenience in the expression of solutions.

\noindent
\textbf{Local and nonlocal NLS (I) hierarchy}:
\begin{equation}
i q_{t_{2l}}=  K^{NLS(I)}_{2l}=  -K_{1,2l}|_{\eqref{nls-hie-reduction1}},~~ l=1,2,\cdots,
\label{nls-hie1}
\end{equation}
reduced from the even-indexed members of the AKNS hierarchy \eqref{akns-hierarchy} with the reduction
\begin{equation}\label{nls-hie-reduction1}
 r(x,t)= \delta q^*(\sigma x,t),~~\delta,\sigma=\pm 1,~~t\to i t,~~x,t \in \mathbb{R},
\end{equation}
where the representative is
\begin{equation}
iq_{t_2}=  K^{NLS(I)}_{2}= q_{xx}- 2\delta q^2 q^*(\sigma x,t).
\label{nls-eq1}
\end{equation}

\noindent
\textbf{Reverse-$t$ NLS(II) hierarchy}:
\begin{equation}
 q_{t_{2l}}=K^{NLS(II)}_{2l}=K_{1,2l}|_{\eqref{nls-hie-reduction2}},~~ l=1,2,\cdots,
\label{nls-hie2}
\end{equation}
reduced from the even-indexed members of the AKNS hierarchy \eqref{akns-hierarchy} with the reduction
\begin{equation}\label{nls-hie-reduction2}
 r(x,t)= \delta q(\sigma x,-t),~~\delta,\sigma=\pm 1,
\end{equation}
where the representative is
\begin{equation}
q_{t_2}=K^{NLS(II)}_{2}=-q_{xx}+ 2\delta q^2 q(\sigma x,-t).
\label{nls-eq2}
\end{equation}

\noindent
\textbf{Local and nonlocal mKdV hierarchy}:
\begin{equation}
q_{t_{2l+1}}=K^{mKdV}_{2l+1}=K_{1,2l+1}|_{\eqref{mkdv-hie-red}},~~ l=0, 1, 2\cdots,
\label{mkdv-hie}
\end{equation}
reduced from the odd-indexed members of the AKNS hierarchy \eqref{akns-hierarchy} with the reduction
\begin{equation}\label{mkdv-hie-red}
 r(x,t)= \delta q(\sigma_1 x,\sigma_2 t),~~\delta,\sigma_i=\pm 1,~~\sigma_1\sigma_2=1%,~~\textcolor[rgb]{1.00,0.00,0.00}{x,t\in \mathbb{C},??}
\end{equation}
where the representative is
\begin{equation}
q_{t_3}=K^{mKdV}_{3}=q_{xxx}-6 \delta q q_x q(\sigma_1 x,\sigma_2 t).
\label{mkdv-eq}
\end{equation}

\noindent
\textbf{Local and nonlocal cmKdV hierarchy}:
\begin{equation}
q_{t_{2l+1}}=K^{cmKdV}_{2l+1}=K_{1,2l+1}|_{\eqref{cmkdv-hie-red}},~~ l=0, 1, 2, \cdots,
\label{cmkdv-hie}
\end{equation}
reduced from the odd-indexed members of the AKNS hierarchy \eqref{akns-hierarchy} with the reduction
\begin{equation}\label{cmkdv-hie-red}
 r(x,t)= \delta q^*(\sigma_1 x,\sigma_2 t),~~\delta,\sigma_i=\pm 1,~~\sigma_1\sigma_2=1,~~x,t\in \mathbb{R},
\end{equation}
where the representative is
\begin{equation}
q_{t_3}=K^{cmKdV}_{3}=q_{xxx}-6 \delta q q_x q^*(\sigma_1 x,\sigma_2 t).
\label{cmkdv-eq}
\end{equation}

\subsection{Double Wronskian solution of the AKNS hierarchy}\label{sec-2-2}

The AKNS hierarchy \eqref{akns-hierarchy} is equivalent to the following recursive form\cite{liu1990,newell1985}
\begin{equation}\label{akns-hierarchy-2}
 \left(
   \begin{array}{c}
     q_{t_{n+1}} \\
     r_{t_{n+1}} \\
   \end{array}
 \right)= L \left(
                      \begin{array}{c}
                        q_{t_{n}} \\
                        r_{t_{n}} \\
                      \end{array}
                    \right),~~n=1,2,\cdots
\end{equation}
with assumption $t_1=x$. When introducing  rational transformation
\begin{equation}\label{akns-transformation}
 q=\frac{h}{f},~~r=-\frac{g}{f},
\end{equation}
Eq.\eqref{akns-hierarchy-2} can be written as  bilinear form
\begin{subequations}\label{akns-bilibear}
\begin{eqnarray}
  &&(D_{t_{n+1}}-D_{x}D_{t_n})g\cdot f =0, \label{akns-bilibear-1}  \\
  &&(D_{t_{n+1}}-D_{x}D_{t_n})f\cdot h =0 , \label{akns-bilibear-2}  \\
  &&D^2_{x}f\cdot f =2gh,\label{akns-bilibear-3}
\end{eqnarray}
\end{subequations}
where $D$ is the Hirota derivative defined by\cite{Hir-PTP-1974}
\begin{equation*}
 D_x^m D_y^n f(x,y)\cdot g(x,y) = (\partial_x- \partial_{x'})^m(\partial_y- \partial_{y'})^n f(x,y)g(x',y')|_{x'=x,y'=y}.
\end{equation*}

System \eqref{akns-bilibear} admits  double Wronskian solution\cite{liu1990,yin2008}
\begin{subequations}\label{anks-solution-fgh}
 \begin{eqnarray}
 && f= W^{(n+1,m+1)}(\varphi,\psi)=|\widehat{\varphi}^{(n)}; \widehat{\psi}^{(m)}| , \label{1-solution-of-akns-f}  \\
 && g= 2^{n+1-m} W^{(n+2,m)}(\varphi,\psi)=2^{n+1-m}|\widehat{\varphi}^{(n+1)}; \widehat{\psi}^{(m-1)}| , \label{1-solution-of-akns-g}\\
 && h= 2^{m+1-n}W^{(n,m+2)}(\varphi,\psi)=2^{m+1-n}|\widehat{\varphi}^{(n-1)}; \widehat{\psi}^{(m+1)}|, \label{1-solution-of-akns-h}
 \end{eqnarray}
\end{subequations}
where $\varphi$ and $\psi$ are two column vectors
\begin{equation}
\varphi=(\varphi_1,\varphi_2,\cdots,\varphi_{n+m+2})^T,~~
\psi=(\psi_1,\psi_2,\cdots,\psi_{n+m+2})^T
\label{phipsi}
\end{equation}
$\widehat{\varphi}^{(n)}$ stands for consecutive columns $(\varphi, \partial_x\varphi, \partial^2_x\varphi,\cdots,\partial^{n}_x\varphi)$,
and $\varphi$ and $\psi$ are
defined by
 \begin{equation}\label{akns-var-psi-A-c-d}
  \varphi=\exp{\Bigl(\frac{1}{2}\sum^{\infty}_{j=1} A^j t_j\Bigr)}C^+, ~~\psi=\exp{\Bigl( -\frac{1}{2}\sum^{\infty}_{j=1} A^j t_j\Bigr)}C^-,
 \end{equation}
with $A=(k_{ij})_{(n+m+2)\times (n+m+2)}$ being an arbitrary constant matrix in $\mathbb{C}_{(n+m+2)\times (n+m+2)}$,
and
\[C^{\pm}=(c_1^{\pm}, c_2^{\pm},\cdots, c_{n+m+2}^{\pm})^T,~~ c^{\pm}_i\in \mathbb{C}.\]

Three comments we have on the above formulation.
Firstly, it can be prove that $A$ and any of its similar matrix lead to same $q$ and $r$.
Secondly, for the odd-indexed members in the hierarchy \eqref{akns-hierarchy},
\begin{equation}
u_{t_{2l+1}}=K_{2l+1},~~ u=(q,r)^T,~~ l=0,1,2,\cdots,
\label{akns-hie-odd}
\end{equation}
their  solutions are given by taking  ($t_1=x$)
\begin{equation}\label{akns-var-psi-A-c-d-odd}
\varphi=\exp{\Bigl(\frac{1}{2}\sum^{\infty}_{j=0} A^{2j+1} t_{2j+1}\Bigr)}C^+, ~~
\psi=\exp{\Bigl( -\frac{1}{2}\sum^{\infty}_{j=0} A^{2j+1} t_{2j+1}\Bigr)}C^-,
\end{equation}
i.e. the terms $ e^{\pm \frac{1}{2} A^{2j} t_{2j}}$ in \eqref{akns-var-psi-A-c-d} are absorbed into $C^{\pm}$.
Finally, for the even-indexed members in the hierarchy \eqref{akns-hierarchy},
\begin{equation}
u_{t_{2l}}=K_{2l},~~ u=(q,r)^T,~~ l=1,2,\cdots,
\label{akns-hie-even}
\end{equation}
their  solutions are given by taking
\begin{equation}\label{akns-var-psi-A-c-d-even}
\varphi=\exp{\Bigl(\frac{1}{2}Ax+\frac{1}{2}\sum^{\infty}_{j=1} A^{2j} t_{2j}\Bigr)}C^+, ~~
\psi=\exp{\Bigl( -\frac{1}{2}Ax-\frac{1}{2}\sum^{\infty}_{j=1} A^{2j} t_{2j}\Bigr)}C^-.
\end{equation}

\section{Reductions for Solutions}\label{section-reduction-of-solution}

In this section, we employ the technique developed in \cite{ChenZ-AML-2017} for the nonlocal NLS equations
to construct solutions for the hierarchies \eqref{kdv-hie}, \eqref{nls-hie1}, \eqref{nls-hie2}, \eqref{mkdv-hie} and \eqref{cmkdv-hie}.
We will design suitable constraints on the pair $(\varphi, \psi)$ in the double Wronskians so that
\eqref{akns-transformation} coincides with the reductions \eqref{kdv-hie-reduction},
\eqref{nls-hie-reduction1}, \eqref{nls-hie-reduction2}, \eqref{mkdv-hie-red} and \eqref{cmkdv-hie-red}, respectively.
Details will be given for the KdV hierarchy and NLS(II) hierarchy.
Results of other hierarchies will  be listed in Sec.\ref{sec-3-3}.

\subsection{Solution to the KdV hierarchy}
For the KdV hierarchy \eqref{kdv-hie}, its Wronskian solution can be derived from the one of the AKNS system \eqref{akns-hie-odd}.
%described as the following theorem.

\begin{theorem}\label{solution-kdvs}
The  KdV hierarchy \eqref{kdv-hie} admits  Wronskian solution
\begin{equation}\label{kdv-r-solu}
 r(x,t)= -\frac{4|\widehat{\varphi}^{(m+2)}; \widehat{\psi}^{(m-1)}|}{|\widehat{\varphi}^{(m+1)}; \widehat{\psi}^{(m)}|},
\end{equation}
where $\varphi$ and $\psi$ are $(2m+3)$-th order column vectors
%$($i.e. $n=m+1$ in \eqref{anks-solution-fgh}$)$
defined by \eqref{akns-var-psi-A-c-d-odd} with $n=m+1$, and obey the constraint
\begin{equation}\label{kdv-psiTphi}
 \psi(x,t)=T \varphi(x,t),~~C^-=TC^+,
\end{equation}
in which $T\in \mathbb{C}_{(2m+3)\times (2m+3)}$ is a constant matrix satisfying
\begin{equation}\label{kdv-at-equation}
AT+TA=0 ,~ T^2=I,
\end{equation}
where $I$ is the $(2m+3)$-th order identity matrix.

\end{theorem}

\begin{proof}

First, starting from \eqref{akns-var-psi-A-c-d-odd}, with assumption
\begin{equation}
AT+TA=0
\label{AT-kdv}
\end{equation}
and $C^-=TC^+$, one can verify that
\begin{align*}
\psi(x,t)=& \exp{\Bigl( -\frac{1}{2}\sum^{\infty}_{j=0} A^{2j+1} t_{2j+1}\Bigr)}C^-  \\
=&  \exp{\Bigl( -\frac{1}{2}\sum^{\infty}_{j=0} (-TAT^{-1})^{2j+1} t_{2j+1}\Bigr)}TC^+  \\
=& \exp{\Bigl( T(\frac{1}{2}\sum^{\infty}_{j=0} A^{2j+1} t_{2j+1})T^{-1}\Bigr)}TC^+  \\
=& T\exp{\Bigl(\frac{1}{2}\sum^{\infty}_{j=0} A^{2j+1} t_{2j+1})\Bigr)}C^+ \\
=& T \varphi(x,t),
\end{align*}
which means  the constraint \eqref{kdv-psiTphi} and assumption \eqref{AT-kdv} are compatible.

Next, we show that under \eqref{kdv-psiTphi} and \eqref{kdv-at-equation},
for $q$ and $r$ defined by \eqref{akns-transformation} with \eqref{anks-solution-fgh},
the reduction $q=\pm1$  holds.
In fact, in addition to \eqref{kdv-psiTphi} we take $n=m+1$ and  assume $T^2=I$, i.e. $|T|=\pm 1$. It then follows that
\begin{align*}
 h= & |\widehat{\varphi}^{(m)}; \widehat{\psi}^{(m+1)}|  = |\widehat{\varphi}^{(m)}; T\widehat{\varphi}^{(m+1)}|
          = |T||T\widehat{\varphi}^{(m)}; \widehat{\varphi}^{(m+1)}|    \\
= & |T||\widehat{\varphi}^{(m+1)};T\widehat{\varphi}^{(m)}| = \pm |\widehat{\varphi}^{(m+1)};T\widehat{\varphi}^{(m)}|= \pm f.
\end{align*}
As a consequence, $q=h/f=|T|=\pm 1$, which  coincides with the reduction \eqref{kdv-hie-reduction}.
Meanwhile, $r$ is expressed by \eqref{kdv-r-solu}.

\end{proof}

\begin{corollary}
As for solutions to \eqref{kdv-at-equation}, we only need to consider the case when $A$ is in its canonical form.
Then, solutions to \eqref{kdv-at-equation} are given below:
\begin{equation}\label{kdv-T-A-C}
 T= \left(
      \begin{array}{ccc}
        \mathbf{0}_{m+1} & 0_{m+1} & \mathbf{I}_{m+1} \\
        0_{m+1}^T & 1 & 0_{m+1}^T \\
        \mathbf{I}_{m+1} & 0_{m+1} & \mathbf{0}_{m+1} \\
      \end{array}
    \right),~~ A= \left(
                   \begin{array}{ccc}
                     \mathbf{K}_{m+1} & 0_{m+1} & \mathbf{0}_{m+1} \\
                     0_{m+1}^T & 0 & 0_{m+1}^T \\
                     \mathbf{0}_{m+1} & 0_{m+1} & -\mathbf{K}_{m+1} \\
                   \end{array}
                 \right),
\end{equation}
where $\mathbf{0}_{m+1}$ is a $(m+1)\times (m+1)$ matrix in which all elements are zeros,
$\mathbf{I}_{m+1}$ is the $(m+1)\times (m+1)$ identity matrix,
$\mathbf{K}_{m+1}\in \mathbb{C}_{(m+1)\times (m+1)}$,
$0_{m+1}$ is a $(m+1)$-th order column vector in which all elements are zeros.
For explicit forms of $\varphi$ and $\psi$, when
\[\mathbf{K}_{m+1}=\mathrm{Diag}(k_1,k_2,\cdots, k_{m+1}),\]
we have
 \begin{align}\label{kdv-varphi-psi}
 &\varphi= (c_1 e^{\xi(k_1)}, c_{2} e^{\xi(k_2)},\cdots,c_{m+1} e^{\xi(k_{m+1})},1,
 d_1 e^{\xi(-k_1)}, d_{2} e^{\xi(-k_2)},\cdots,d_{m+1} e^{\xi(-k_{m+1})})^T,
 \end{align}
where
\begin{equation}\label{kdv-scatter}
\xi(k_i)=\frac{1}{2}k_i x+\frac{1}{2}\sum^{\infty}_{j=1}k_i^{2j+1}t_{2j+1}.
\end{equation}
When  $\mathbf{K}_{m+1}$ is a $(m+1)\times (m+1)$ Jordan matrix w.r.t. %$k_1$,
\begin{equation}
J_{m+1}(k)=\left(
      \begin{array}{cccc}
        k & 0 & \cdots & 0  \\
        1& k& \cdots& 0 \\
        \vdots &\ddots&\ddots& \vdots\\
        0 & \cdots &1&k
      \end{array}\right)_{(m+1)\times (m+1)},
\label{Jordan}
\end{equation}
we have
 \begin{equation}\label{kdv-varphi-jord}
 \varphi= \left(c e^{\xi(k)}, \frac{\partial_{k}}{1!}(c e^{\xi(k)}),\cdots, \frac{\partial^m_{k}}{m!} (c e^{\xi(k)}),1,
 d e^{\xi(-k)}, \frac{\partial_{k}}{1!}(d e^{\xi(-k)}),\cdots,\frac{\partial^m_{k}}{m!}(d e^{\xi(-k)})\right)^T.
 \end{equation}
Since $q=h/f=|T|$ and $|T|=(-1)^{m+1}$ where $T$ is defined in \eqref{kdv-T-A-C},
to obtain solutions to the KdV hierarchy \eqref{kdv-hie} corresponding to reduction $q=1$, we can replace \eqref{kdv-r-solu} with
\begin{equation}\label{kdv-r-solu-T}
 r(x,t)= (-1)^m \frac{4|\widehat{\varphi}^{(m+2)}; \widehat{\psi}^{(m-1)}|}{|\widehat{\varphi}^{(m+1)}; \widehat{\psi}^{(m)}|}.
\end{equation}

\end{corollary}

As for solitons, \eqref{kdv-r-solu-T} with \eqref{kdv-varphi-psi}
gives $(m+1)$-soliton solution for the KdV hierarchy corresponding to reduction $q=1$.
The simplest one is 1-soliton solution ($m=0$),
\begin{equation*}
%\label{kdv-one-soli-solu}
 r= -\frac{k^2}{2} \mathrm{sech}^2\Bigl(\frac{1}{2}k x +\frac{1}{2}k^3 t_3+\frac{1}{2}\xi^{(0)}\Bigr),~~ \xi^{(0)}=2\ln c \in \mathbb{R},
\end{equation*}
where we have taken $C^+=(c,1,-1/c)^T$.
Solutions corresponding to Jordan block \eqref{Jordan} are usually called multiple-pole solutions.

Note that \eqref{kdv-r-solu-T} provide a soliton solution to the KdV equation in double Wronskian form.
However, it is not clear how this form is related to the solutions in single Wronskian form (cf.\cite{Zha-2006}).

\subsection{Solutions of the NLS(II) hierarchy \eqref{nls-hie2}}

\subsubsection{Solutions}

For the NLS(II) hierarchy \eqref{nls-hie2}, we first present results and then prove them.

\begin{theorem}\label{solution-nls2s}

The  NLS(II) hierarchy \eqref{nls-hie2} has solution
\begin{equation}\label{nls2-q-solu}
 q(x)= \frac{2|\widehat{\varphi}^{(n-1)}; \widehat{\psi}^{(n+1)}|}{|\widehat{\varphi}^{(n)}; \widehat{\psi}^{(n)}|},
\end{equation}
where $\varphi$ and $\psi$ are $2(n+1)$-th order column vectors $($i.e. m=n in \eqref{anks-solution-fgh}$)$ defined by \eqref{akns-var-psi-A-c-d-even},
and obey the constraint
\begin{equation}\label{nls2-psiTphi}
 \psi(x,t)=T \varphi(\sigma x,-t),
\end{equation}
in which $T\in \mathbb{C}_{2(n+1)\times 2(n+1)}$ is a constant matrix satisfying
\begin{equation}\label{nls2-at-equation}
AT+\sigma TA=0 ,~~ T^2=\delta \sigma I,~C^-=TC^+.
\end{equation}
\end{theorem}

\begin{proof}

Similar to the KdV case, we assume
\begin{equation}\label{nls2-at-equation-assu}
AT+\sigma TA=0 , ~C^-=TC^+.
\end{equation}
Then we have
\begin{subequations}
\begin{eqnarray}
&&\psi(x,t)= \exp{\Bigl( -\frac{1}{2}Ax-\frac{1}{2}\sum^{\infty}_{j=1} A^{2j} t_{2j}\Bigr)}C^-  \nonumber  \\
&&~~~~~~~~~=  \exp{\Bigl( -\frac{1}{2}(-\sigma T AT^{-1})x-\frac{1}{2}\sum^{\infty}_{j=1} (-\sigma T AT^{-1})^{2j} t_{2j}\Bigr)}TC^+  \nonumber  \\
&&~~~~~~~~~=  \exp{\Bigl( \frac{1}{2}( T AT^{-1})\sigma x + \frac{1}{2}\sum^{\infty}_{j=1} ( T AT^{-1})^{2j}(- t_{2j})\Bigr)}TC^+  \nonumber  \\
&&~~~~~~~~~=  T\exp{\Bigl( \frac{1}{2} A \sigma x + \frac{1}{2} \sum^{\infty}_{j=1} A^{2j}(- t_{2j})\Bigr )}C^+  \nonumber  \\
&&~~~~~~~~~= T \varphi(\sigma x,-t), \nonumber
\end{eqnarray}
\end{subequations}
which means  the constraint \eqref{nls2-psiTphi} and assumption \eqref{nls2-at-equation-assu} are compatible.
For convenience we introduce notation
\begin{equation*}
 \widehat{\varphi}^{(s)}(ax)_{[bx]}=(\varphi(ax),\partial_{bx}\varphi(ax),\cdots, \partial^{s}_{bx}\varphi(ax))_{2(n+1)\times 2(n+1)},
\end{equation*}
where $\varphi(x)=(\varphi_1(x),\varphi_2(x),\cdots,\varphi_{2(n+1)}(x))^T$ is a $2(n+1)$-th order column vector.
Under the constraint \eqref{nls2-psiTphi}, the double Wronskians $f,g,h$ in \eqref{anks-solution-fgh} are rewritten as
\begin{subequations}\label{nls2-fgh}
\begin{eqnarray}
&&f(x,t)= |\widehat{\varphi}^{(n)}; \widehat{\psi}^{(n)}|= |\widehat{\varphi}^{(n)}(x,t)_{[x]};
T\widehat{\varphi}^{(n)}(\sigma x,-t)_{[x]}|,  \label{nls2-fgh-f}  \\
&&g(x,t)= 2|\widehat{\varphi}^{(n+1)}; \widehat{\psi}^{(n-1)}|=  2|\widehat{\varphi}^{(n+1)}(x,t)_{[x]};
T\widehat{\varphi}^{(n-1)}(\sigma x,-t)_{[x]}|,  \label{nls2-fgh-g}  \\
&&h(x,t)= 2|\widehat{\varphi}^{(n-1)}; \widehat{\psi}^{(n+1)}|=2 |\widehat{\varphi}^{(n-1)}(x,t)_{[x]};
T\widehat{\varphi}^{(n+1)}(\sigma x,-t)_{[x]}|. \label{nls2-fgh-h}
\end{eqnarray}
\end{subequations}
With a further assumption $T^2=\delta \sigma I$,
we have
\begin{subequations}
\begin{eqnarray}
&&f(\sigma x,-t)= |\widehat{\varphi}^{(n)}(\sigma x,-t)_{[\sigma x]}; T\widehat{\varphi}^{(n)}(\sigma^2 x,t)_{[\sigma x]}|   \nonumber  \\
&&~~~~~~~~~~~~ = (-1)^{(n+1)^2}|T| | \widehat{\varphi}^{(n)}(x,t)_{[x]};T^{-1} \widehat{\varphi}^{(n)}(\sigma x,-t)_{[x]}|   \nonumber  \\
&&~~~~~~~~~~~~ = (-1)^{(n+1)^2}(\delta \sigma )^{n+1}|T| | \widehat{\varphi}^{(n)}(x,t)_{[x]};T \widehat{\varphi}^{(n)}(\sigma x,-t)_{[x]}|   \nonumber  \\
&&~~~~~~~~~~~~ = (-1)^{(n+1)^2}(\delta \sigma )^{n+1}|T| f(x,t) \nonumber
\end{eqnarray}
\end{subequations}
and
\begin{subequations}
\begin{eqnarray}
&&h(\sigma x,-t)=2 |\widehat{\varphi}^{(n-1)}(\sigma x,-t)_{[\sigma x]}; T\widehat{\varphi}^{(n+1)}(\sigma^2 x,t)_{[\sigma x]}|   \nonumber  \\
&&~~~~~~~~~~~~ = 2(-1)^{n(n+2)}\sigma^{2n+1}|T| | \widehat{\varphi}^{(n+1)}(x,t)_{[x]};T^{-1} \widehat{\varphi}^{(n-1)}(\sigma x,-t)_{[x]}|   \nonumber  \\
&&~~~~~~~~~~~~ =2(-1)^{n(n+2)}\sigma^{2n+1}(\delta \sigma)^{n}|T||\widehat{\varphi}^{(n+1)}(x,t)_{[x]};T\widehat{\varphi}^{(n-1)}(\sigma x,-t)_{[x]}| \nonumber\\
&&~~~~~~~~~~~~ = (-1)^{n(n+2)}\sigma(\delta \sigma)^{n}|T| g(x,t), \nonumber
\end{eqnarray}
\end{subequations}
which immediately yields relation
\begin{equation*}
\delta q(\sigma x,-t)= \delta\frac{h(\sigma x,-t)}{f(\sigma x,-t)}
= \delta\frac{(-1)^{n(n+2)}\sigma(\delta \sigma)^{n}|T| g(x,t)}{(-1)^{(n+1)^2}(\delta \sigma )^{n+1}|T| f(x,t)}=-\frac{g(x,t)}{f(x,t)}=r(x,t),
\end{equation*}
which coincides with the reduction \eqref{nls-hie-reduction2} for the NLS(II) hierarchy \eqref{nls-hie2}.
Thus we finish the proof.

\end{proof}

Solutions to the condition equation \eqref{nls2-at-equation} with respect to $(\sigma, \delta)$ are given below.
Suppose $T$ and $A$ are block matrices
\begin{equation}\label{nls2-example-tac}
 T=\left(
     \begin{array}{cc}
       T_1 & T_2 \\
       T_3 & T_4 \\
     \end{array}
   \right),~~A=\left(
                 \begin{array}{cc}
                   K_1 & \mathbf{0} \\
                   \mathbf{0} & K_4 \\
                 \end{array}
               \right),
\end{equation}
where $T_i$ and $A_i$ are $(n+1)\times (n+1)$ matrices.
Then, solutions to \eqref{nls2-at-equation} are

\begin{table}[H]
\begin{center}
\begin{tabular}{|c|c|c|}
\hline
   $(\sigma, \delta)$    & $T$ &  $A$    \\
\hline
   $(1,-1)$              & $T_1=T_4=\mathbf{0}_{n+1}$, $T_3=-T_2 =\mathbf{I}_{n+1}$ & $ K_1=-K_4=\mathbf{K}_{n+1}$ \\
\hline
   $(1,1)$              & $T_1=T_4=\mathbf{0}_{n+1}$, $T_3=T_2 =\mathbf{I}_{n+1}$ & $ K_1=-K_4=\mathbf{K}_{n+1}$ \\
\hline
   $(-1,-1)$              & $T_1=-T_4=\mathbf{I}_{n+1}$, $T_3=T_2 =\mathbf{0}_{n+1}$ & $ K_1=-K_4=\mathbf{K}_{n+1}$ \\
\hline
   $(-1,1)$              & $T_1=-T_4=i \mathbf{I}_{n+1}$, $T_3=T_2 =\mathbf{0}_{n+1}$ & $ K_1=-K_4=\mathbf{K}_{n+1}$ \\
\hline
\end{tabular}
\caption{$T$ and $A$ for the NLS(II) hierarchy}
\label{Tab-1}
\end{center}
\end{table}

In case of $\mathbf{K}_{n+1}$ being a diagonal matrix
\begin{equation}
\mathbf{K}_{n+1}=\mathrm{Diag}(k_1,k_2,\cdots, k_{n+1}),
\label{K-n}
\end{equation}
we have
\begin{equation}\label{nls-II-varphi-psi}
\varphi= (c_1 e^{\zeta(k_1)}, c_2 e^{\zeta(k_2)},\cdots,c_{n+1} e^{\zeta(k_{n+1})},
 d_1 e^{\zeta(-k_1)},  d_2 e^{\zeta(-k_2)},\cdots, d_{n+1} e^{\zeta(-k_{n+1})})^T,%~~ c_j, d_j\in \mathbb{C},
\end{equation}
%\end{subequations}
where
\begin{equation}\label{nle-II-scatter}
\zeta(k_i)=\frac{1}{2}k_i x+\frac{1}{2}\sum^{\infty}_{j=1}k_i^{2j}t_{2j}.
\end{equation}
When  $\mathbf{K}_{n+1}$ is a $(n+1)\times (n+1)$ Jordan matrix w.r.t. $J_{n+1}(k)$ defined as in \eqref{Jordan},
we have
 \begin{equation}\label{nls-I-varphi-jord}
 \varphi= \left(c e^{\zeta(k)}, \frac{\partial_{k}}{1!}(c e^{\zeta(k)}),\cdots, \frac{\partial^n_{k}}{n!} (c e^{\zeta(k)}),
 d e^{\zeta(-k)}, \frac{\partial_{k}}{1!}(d e^{\zeta(-k)}),\cdots,\frac{\partial^n_{k}}{n!}(d e^{\zeta(-k)})\right)^T.
 \end{equation}

For  equation \eqref{nls-eq2} with different $(\sigma,\delta)$, their one-soliton solutions\footnote{Here we still call them
soliton solutions regardless their traveling behaviors.} (1SS) are respectively
\begin{subequations}\label{iss-nls-2}
\begin{align}
&  q_{\sigma=1, \delta=-1} =-\frac{2k c d\,e^{-k^2 t}}{c^2 e^{kx}+ d^2 e^{-kx}}, \label{1ss-nls-2(1-1)} \\
&  q_{\sigma=1, \delta=1} =\frac{2k c d\,e^{-k^2 t}}{c^2 e^{kx}- d^2 e^{-kx}},\label{1ss-6}\\
&  q_{\sigma=-1, \delta=1} =\frac{2ik e^{-k^2 t}}{e^{kx}+e^{-kx}},\label{1ss-7}\\
& q_{\sigma=-1, \delta=-1} =\frac{2k e^{-k^2 t}}{e^{kx}+e^{-kx}},\label{1ss-nls-2(1-2)}
\end{align}
\end{subequations}
where $k=k_1, t=t_2$ and $C^+=(c,d)^T$.

\begin{remark}\label{rem-1}
By observation of equations in the hierarchy \eqref{nls-hie2} one can find the following.
On the level of equations,
\eqref{nls-hie2} with $(\sigma, \delta)=(\pm1,1)$ and \eqref{nls-hie2} with $(\sigma, \delta)=(\pm1,-1)$ can be transformed from each other
by simply taking $q\to iq$.
Such a transformation can be also seen on the level of solutions that we present through \eqref{nls2-q-solu} with choices listed in Table \ref{Tab-1}.
First, for $q_{(\sigma, \delta)=(1,-1)}$, due to arbitrariness  of vector $C^{+}$,
if we replace $C^{+}$ by $\mathrm{Diag}(\mathbf{I}_{n+1}, i \mathbf{I}_{n+1})C^{+}$,
then, $q_{(\sigma, \delta)=(1,-1)}$ defined by \eqref{nls2-q-solu}  will be changed to $iq_{(\sigma, \delta)=(1,-1)}$,
which gives a solution to \eqref{nls-hie2} with $(\sigma, \delta)=(1,1)$.
Second, for $q$ given in the form \eqref{nls2-q-solu}, from the line $(\sigma, \delta)=(-1,-1)$ and line $(\sigma, \delta)=(-1,1)$ in Table \ref{Tab-1},
one can immediately find the relation $q_{(\sigma, \delta)=(-1,-1)}=-iq_{(\sigma, \delta)=(-1,1)}$.
One-soliton solutions listed in \eqref{iss-nls-2}  also demonstrate the above connections as examples.
In conclusion, with these connections we only need to consider the cases $(\sigma, \delta)=(1,-1)$ and $(\sigma, \delta)=(-1,-1)$.
\end{remark}

\begin{remark}\label{rem-2}
Due to the special block structure of the matrix $T$ for the cases of $(\sigma,\delta)=(-1,\pm 1)$ in Table \ref{Tab-1},
it is easy to find that both $C^{\pm}$ can be gauged to be $(1,1,\cdots, 1)^T$.
This means the solutions we obtain for the cases $(\sigma,\delta)=(-1,\pm 1)$
are independent of phase parameters $C^{\pm}$, i.e. the initial phase has always to be $0$.
\end{remark}

\subsubsection{Dynamics of  $q_{(\sigma, \delta)=(1,-1)}$}\label{sec-3-2-2}

In the following let us investigate dynamics of $q_{(\sigma, \delta)=(1,-1)}$ governed by equation
\begin{equation}
q_t(x,t)+q_{xx}(x,t) + 2 q^2(x,t)q(x,-t)=0.
\label{nls-2-eq}
\end{equation}
The 1SS of this equation is given by \eqref{1ss-nls-2(1-1)}, which is rewritten as
\begin{equation}
q(x,t) =-k e^{-k^2 t} \, \mathrm{sech} \Bigl( kx+\ln \frac{c}{d}\Bigr). \label{1ss-nls-2(1-1)-r}
\end{equation}
This is a stationary wave with an initial phase $\ln\frac{c}{d}$ and an amplitude that exponentially decreases with time,
as depicted in Fig.\ref{fig-1-1}
It is interesting that $q(x,t)q(x,-t)=k^2\mathrm{sech}^2 \bigl( kx+\ln \frac{c}{d}\bigr)$
provides a conserved density for equation \eqref{nls-2-eq}.
The conservation property is illustrated in Fig.\ref{fig-1-2}.

%============================  Fig ===================================
\begin{figure}[!h]
\centering \subfigure[]
{\label{fig-1-1} %% label for first subfigure
\includegraphics[width=2.5in]{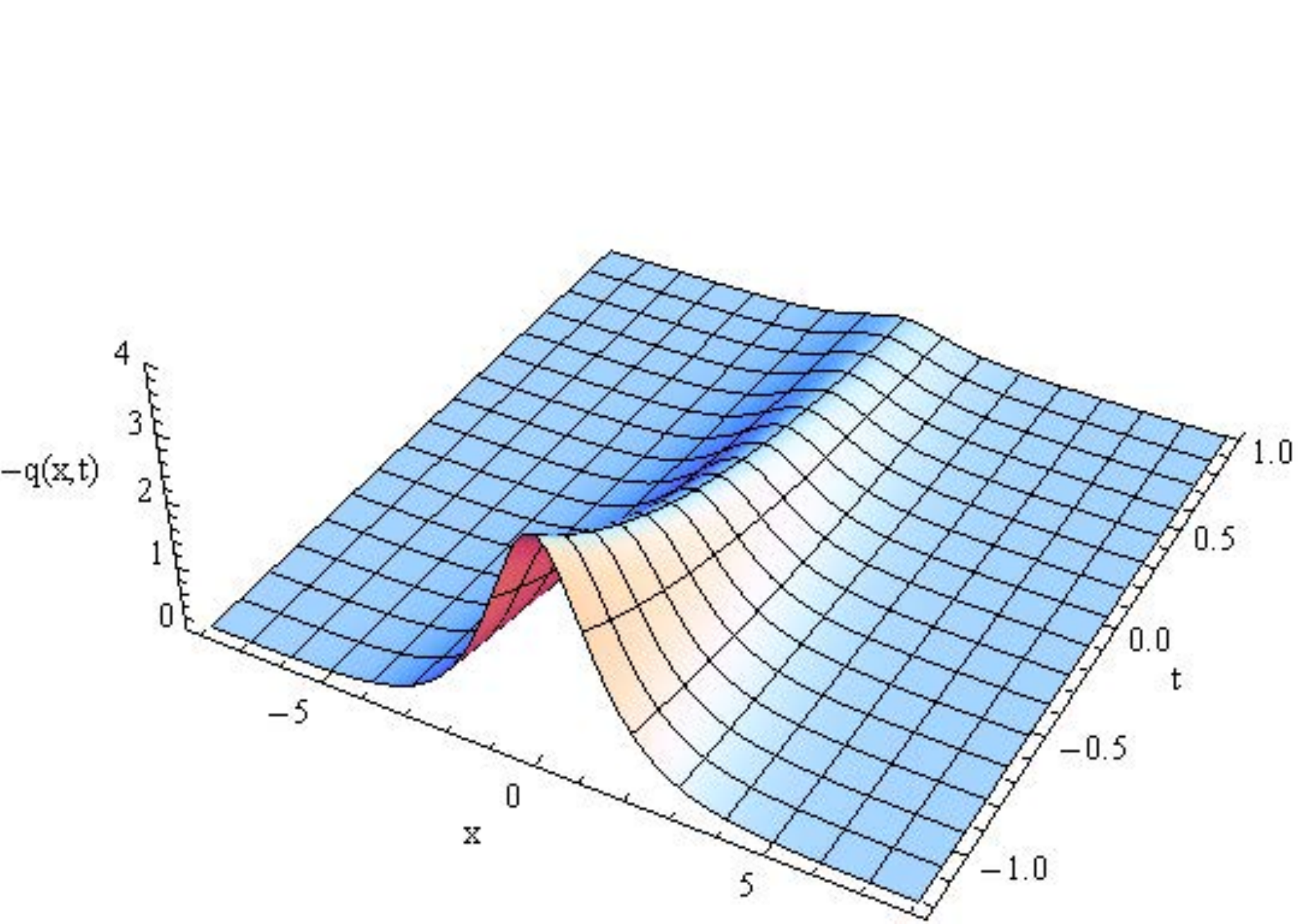}
}~~~ \subfigure[]{
\label{fig-1-2} %% label for second subfigure
\includegraphics[width=2.5in]{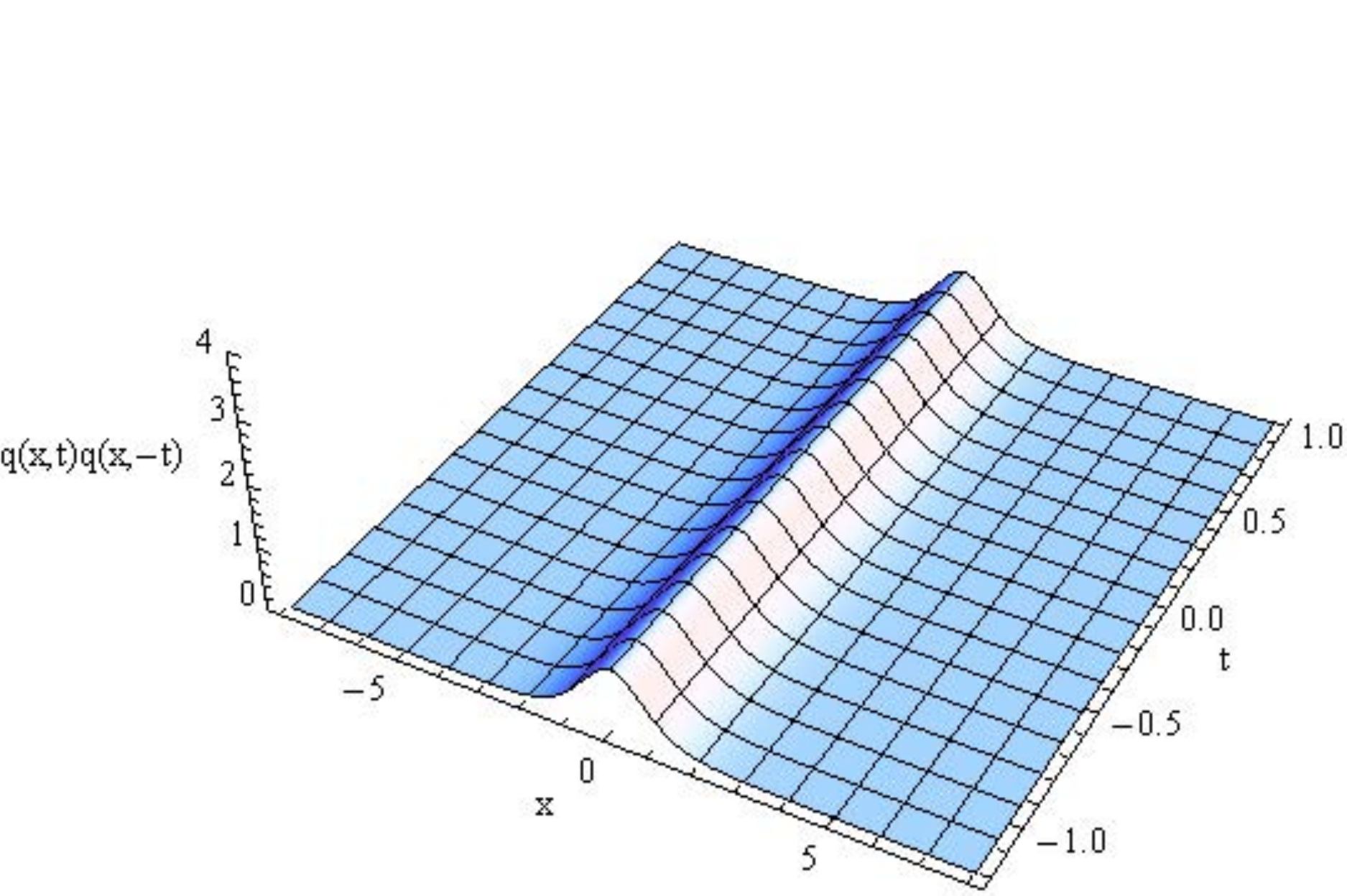}
}\\
\caption{(a). Shape and motion of 1SS \eqref{1ss-nls-2(1-1)-r} for equation \eqref{nls-2-eq},
in which $k=1, c=d=1$.
(b). Shape and motion of $-q(x,t)q(-x,t)$ where $q(x,t)$ is depicted in (a).
\label{Fig-1}}
\end{figure}
%==========================================

From Theorem \ref{solution-nls2s}, 2SS of Eq.\eqref{nls-2-eq} is given by
\begin{equation}\label{nls2-2ss}
q(x,t)= \frac{A}{B}
\end{equation}
where
\begin{equation*}
 A = 2(k_1^2 - k_2^2) [c_1 d_1 k_1 e^{k_2^2 t + k_1 x} (d_2^2 + c_2^2 e^{2 k_2 x}) -
   c_2 d_2 k_2 e^{k_1^2 t + k_2 x} (d_1^2+    c_1^2  e^{2 k_1 x })]
\end{equation*}
and
\begin{align*}
 B=&4 c_1 c_2 d_1 d_2 k_1 k_2 e^{(k_1 + k_2) x} (e^{2 k_1^2 t} + e^{2 k_2^2 t}) \\
 & -e^{(k_1^2 + k_2^2) t}[ c_1^2 e^{2 k_1 x}[c_2^2 (k_1 - k_2)^2 e^{2 k_2 x}+ d_2^2 (k_1 + k_2)^2]
 +d_1^2 [d_2^2 (k_1 - k_2)^2 + c_2^2(k_1 + k_2)^2 e^{2 k_2 x}]].
\end{align*}
Note that 2SS \eqref{nls2-2ss} is non-singular if $c_1 c_2 d_1 d_2 k_1 k_2 < 0$.
Although 1SS is always stationary, 2SS does exhibit interaction.
%============================  Fig ===================================
\begin{figure}[!h]
\centering \subfigure[]
{\label{fig-2-1} %% label for first subfigure
\includegraphics[width=2.5in]{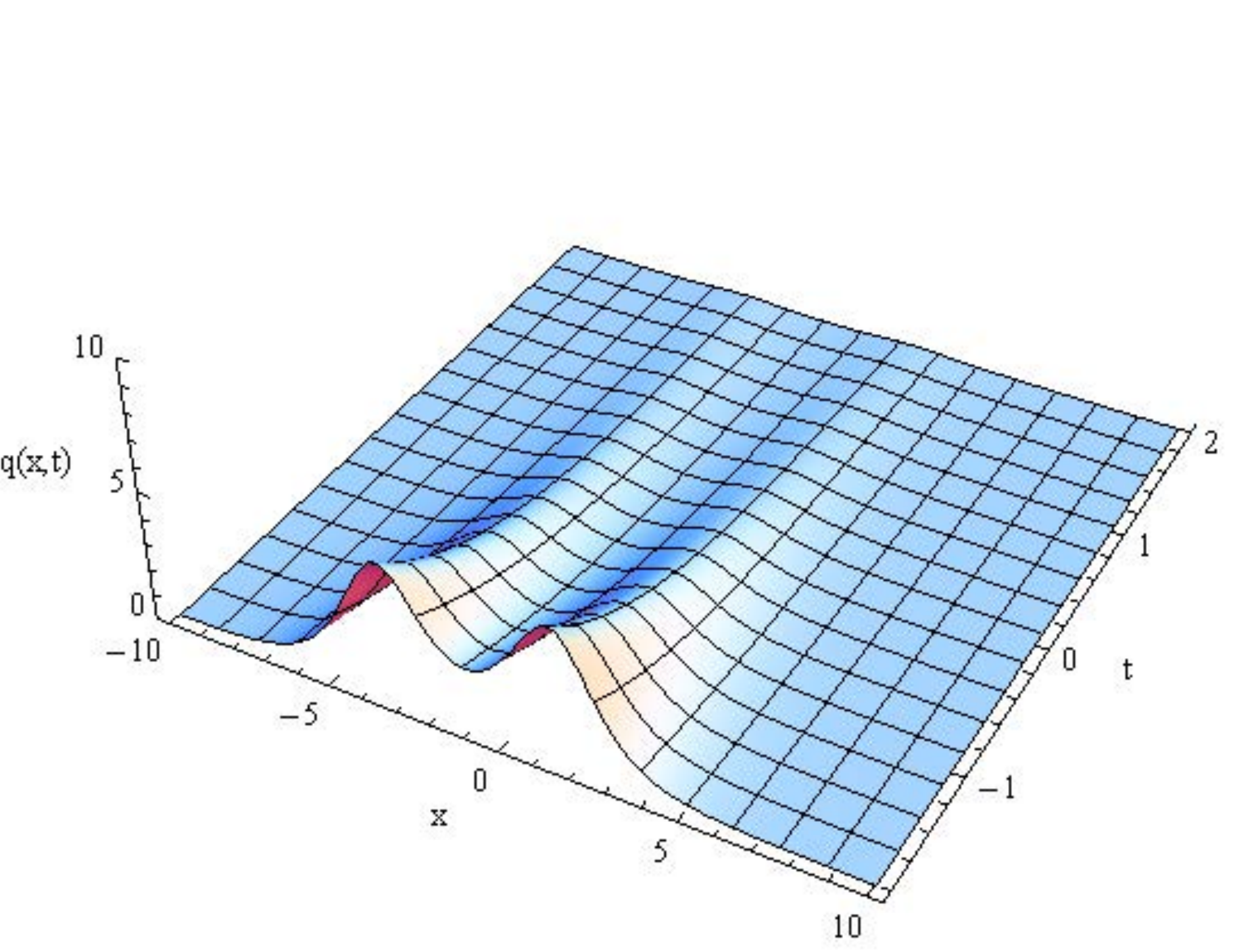}}~~~
\subfigure[]{
\label{fig-2-2} %% label for second subfigure
\includegraphics[width=2.5in]{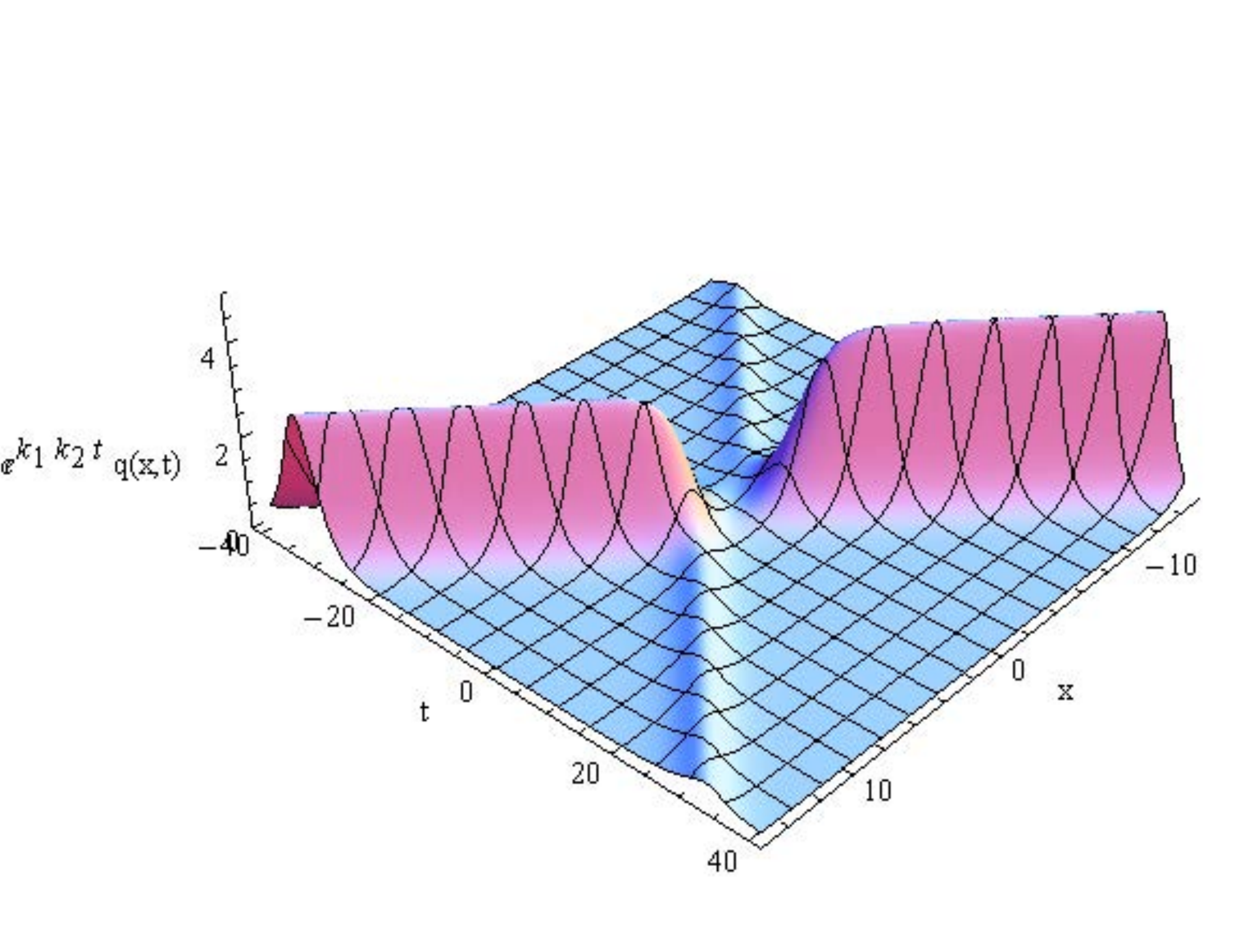}
}
\\
\caption{(a). Shape and motion of 2SS \eqref{nls2-2ss} for equation \eqref{nls-2-eq},
in which $k_1=0.8, k_2=1.2, c_1=2.5, c_2=4, d_1=1, d_2=-1$.
(b). Shape and motion of $e^{k_1k_2 t}q(x,t)$ where $q(x,t)$ is depicted in (a).
\label{Fig-2}}
\end{figure}
%==========================================

In Fig.\ref{fig-2-1} it seems that two solitons are apart from each other.
However, not only do they interact, but also trajectories of two involved solitons
are symmetric.
Fig.\ref{fig-2-1} shows two solitons with exponentially decreasing amplitudes.
We find that in the case $0<k_1<k_2$, $q(x,t)$ multiplied by $e^{k_1k_2t}$, i.e. $e^{k_1k_2t}q(x,t)$,
provides solitons with even amplitudes, as depicted in Fig.\ref{fig-2-2}.
This fact indicates that when $0<k_1<k_2$, both solitons in $q(x,t)$ decrease with $e^{-k_1k_2t}$,
as $e^{k_1k_2t}$ provides a compensation.

With the help of the compensation factor  $e^{k_1k_2t}$  we make asymptotic analysis for $t$ going to infinity in the case of  $0<k_1<k_2$.
To do that, we consider $e^{k_1k_2t}q(x,t)$ respectively along lines
\begin{equation}
x+(k_2-k_1)t=y,~~ x-(k_2-k_1)t=z,
\label{yz}
\end{equation}
where $y$ and $z$ can also be viewed as new variables.
Then we rewrite $e^{k_1k_2t}q(x,t)$ in coordinate frame $\{y,t\}$,
fix $y$ and let $t\to \pm\infty$. We find
\begin{equation}
e^{k_1k_2t}q(x,t)\sim \left\{
        \begin{array}{ll}
         \frac{2k_1(k_2^2-k_1^2)c_1d_2}{(k_2-k_1)^2d_1d_2e^{-k_1 y}-4k_1k_2c_1c_2 e^{k_2 y}},& t\to \infty,\\
         \frac{2k_2(k_2^2-k_1^2)c_1d_2}{4k_1k_2d_1d_2 e^{-k_1 y}-(k_2-k_1)^2c_1c_2e^{k_2 y}},& t\to -\infty.
        \end{array}
        \right.
\end{equation}
It indicates that
along the line $y= \mathrm{constant}$, there is a soliton;
from $t=-\infty$ to $t=\infty$, the soliton gains a shift
\begin{equation}
\Delta=\frac{2}{k_1+k_2} \ln \Bigl(\frac{k_2-k_1}{4k_1k_2}\Bigr)\label{shift}
\end{equation}
and its amplitude changes from
\[A_{-}^{(y)}=\frac{d_2k_1k_2\, 2^{\frac{k_1-k_2}{k_2+k_1}} \Bigl(-\frac{c_1c_2(k_1-k_2)^2}{d_1d_2k_1^2}\Bigr)^{\frac{k_2}{k_1+k_2}}}{c_2(k_1-k_2)}\]
to
\[A_{+}^{(y)}=\frac{c_1k_1k_2\, 2^{\frac{k_2-k_1}{k_2+k_1}} \Bigl(-\frac{d_1d_2(k_1-k_2)^2}{c_1c_2k_2^2}\Bigr)^{\frac{k_1}{k_1+k_2}}}{d_1(k_2-k_1)};\]
the soliton itself is not symmetric in terms of shape because it decays by $e^{k_2(k_2-k_1)t}$ when $t\to -\infty$
and  by $e^{-k_1(k_2-k_1)t}$ when $t\to \infty$.
In a similar way we rewrite $e^{k_1k_2t}q(x,t)$ in coordinate frame $\{z,t\}$,
fix $z$ and let $t\to \pm\infty$. It follows that
\begin{equation}
e^{k_1k_2t}q(x,t)\sim \left\{
        \begin{array}{ll}
         \frac{2k_1(k_2^2-k_1^2)c_2d_1}{(k_2-k_1)^2c_1c_2e^{k_1 z}-4k_1k_2d_1d_2 e^{-k_2 z}},& t\to \infty,\\
         \frac{2k_2(k_2^2-k_1^2)c_2d_1}{4k_1k_2c_1c_2 e^{k_1 z}-(k_2-k_1)^2d_1d_2e^{-k_2 z}},& t\to -\infty.
        \end{array}
        \right.
\end{equation}
In this turn along the line $z= \mathrm{constant}$, there is a soliton;
from $t=-\infty$ to $t=\infty$, the soliton gains a shift \eqref{shift}
and its amplitude changes from
\[A_{-}^{(z)}=\frac{c_2k_1k_2\, 2^{\frac{k_1-k_2}{k_2+k_1}} \Bigl(-\frac{d_1d_2(k_1-k_2)^2}{c_1c_2k_1^2}\Bigr)^{\frac{k_2}{k_1+k_2}}}{d_2(k_1-k_2)}\]
to
\[A_{+}^{(z)}=\frac{d_1k_1k_2\, 2^{\frac{k_2-k_1}{k_2+k_1}} \Bigl(-\frac{c_1c_2(k_1-k_2)^2}{d_1d_2k_2^2}\Bigr)^{\frac{k_1}{k_1+k_2}}}{c_1(k_2-k_1)};\]
the soliton itself is not symmetric in terms of shape.
Note that the lines $y=\mathrm{constant }$ and $z=\mathrm{constant }$ in \eqref{yz}
provide two lines that are completely symmetric in $\{x,t\}$ plane,
the trajectories of two solitons are therefore completely symmetric as well,
which is a distinguished interaction behavior
that is different from usual soliton systems.
Another interesting feather is thatamplitude of each soliton in 2SS is asymptotically related to not only $\{k_i\}$ but also the
phase parameters $\{c_i,d_i\}$.

\subsubsection{Dynamics of  $q_{(\sigma, \delta)=(-1,-1)}$}\label{sec-3-2-3}

Let us quickly go through dynamics of $q_{(\sigma, \delta)=(-1,-1)}$ governed by equation
\begin{equation}
q_t(x,t)+q_{xx}(x,t) + 2 q^2(x,t)q(-x,-t)=0.
\label{nls-2-eq-2}
\end{equation}
Its 1SS \eqref{1ss-nls-2(1-2)} is rewritten as
$
q(x,t) = k e^{-k^2 t} \, \mathrm{sech} kx, $
which is a stationary wave but without any initial phase. Its amplitude also exponentially decreases with time.

2SS is written as
\begin{equation}\label{nls2-2ss-2}
q(x,t)= \frac{A'}{B'},
\end{equation}
where
\[A'=\frac{1}{2}(k_2^2 - k_1^2)e^{-(k_1 + k_2) x} [e^{k_2^2 t + k_1 x} (1 + e^{2 k_2 x}) k_1 - e^{k_1^2 t + k_2 x} (1+e^{2 k_1 x }) k_2]\]
and
\begin{equation*}
B'= (e^{2 k_1^2 t} + e^{2 k_2^2 t}) k_1 k_2 +
    e^{(k_1^2 + k_2^2) t} [2 k_1 k_2 \sinh(k_1 x) \sinh(k_2 x)-(k_1^2 + k_2^2) \cosh(k_1 x) \cosh(k_2 x)].
\end{equation*}
The 2SS \eqref{nls2-2ss-2} is non-singular if $k_1 k_2 < 0$.
Similar to the previous case, although 1SS is always stationary, 2SS does exhibit interaction, as depicted in Fig.\ref{Fig-3}.
%============================  Fig ===================================
\begin{figure}[!h]
\centering \subfigure[]
{\label{fig-3-1} %% label for first subfigure
\includegraphics[width=2.5in]{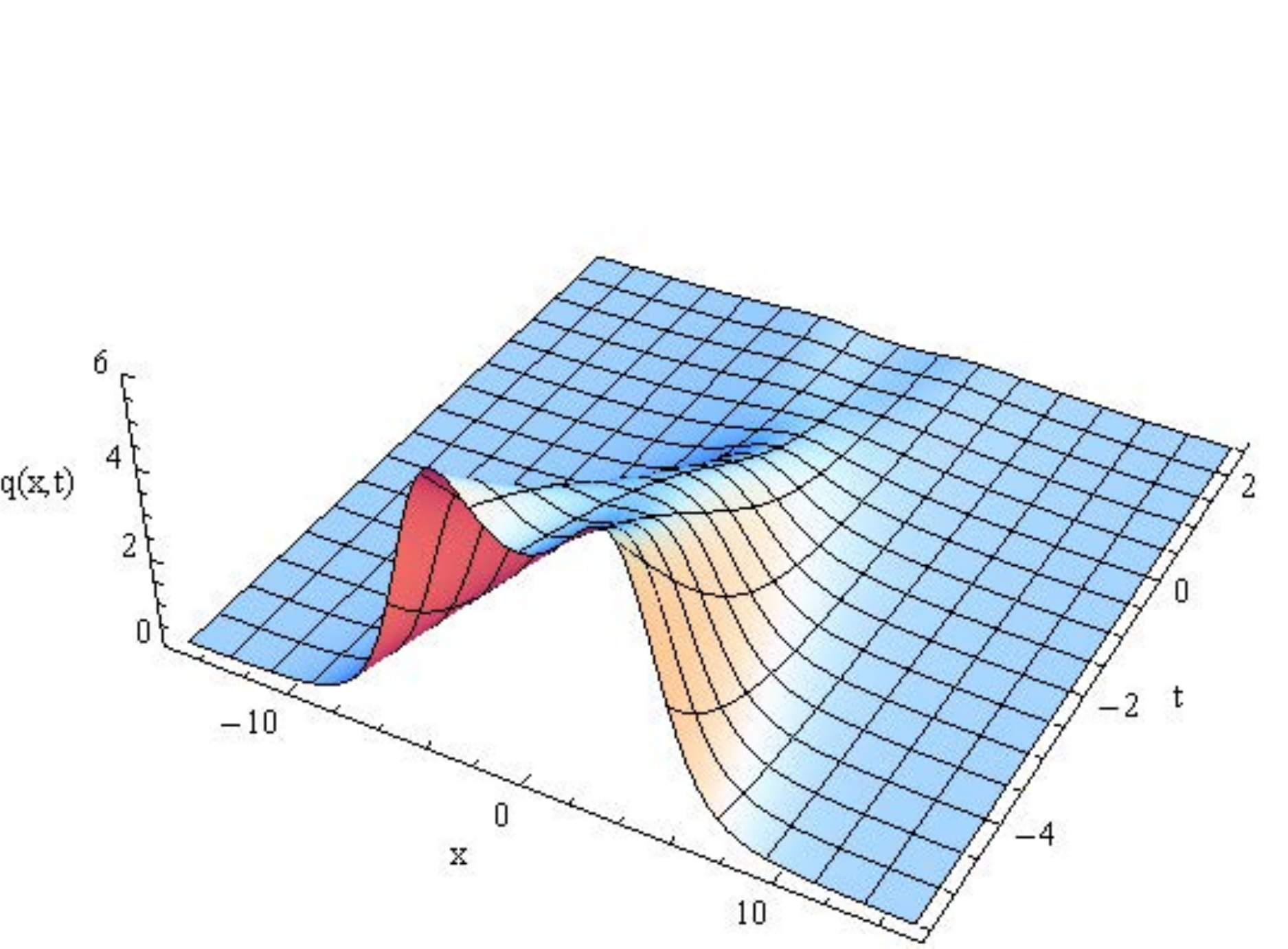}}~~~
\subfigure[]{
\label{fig-3-2} %% label for second subfigure
\includegraphics[width=2.5in]{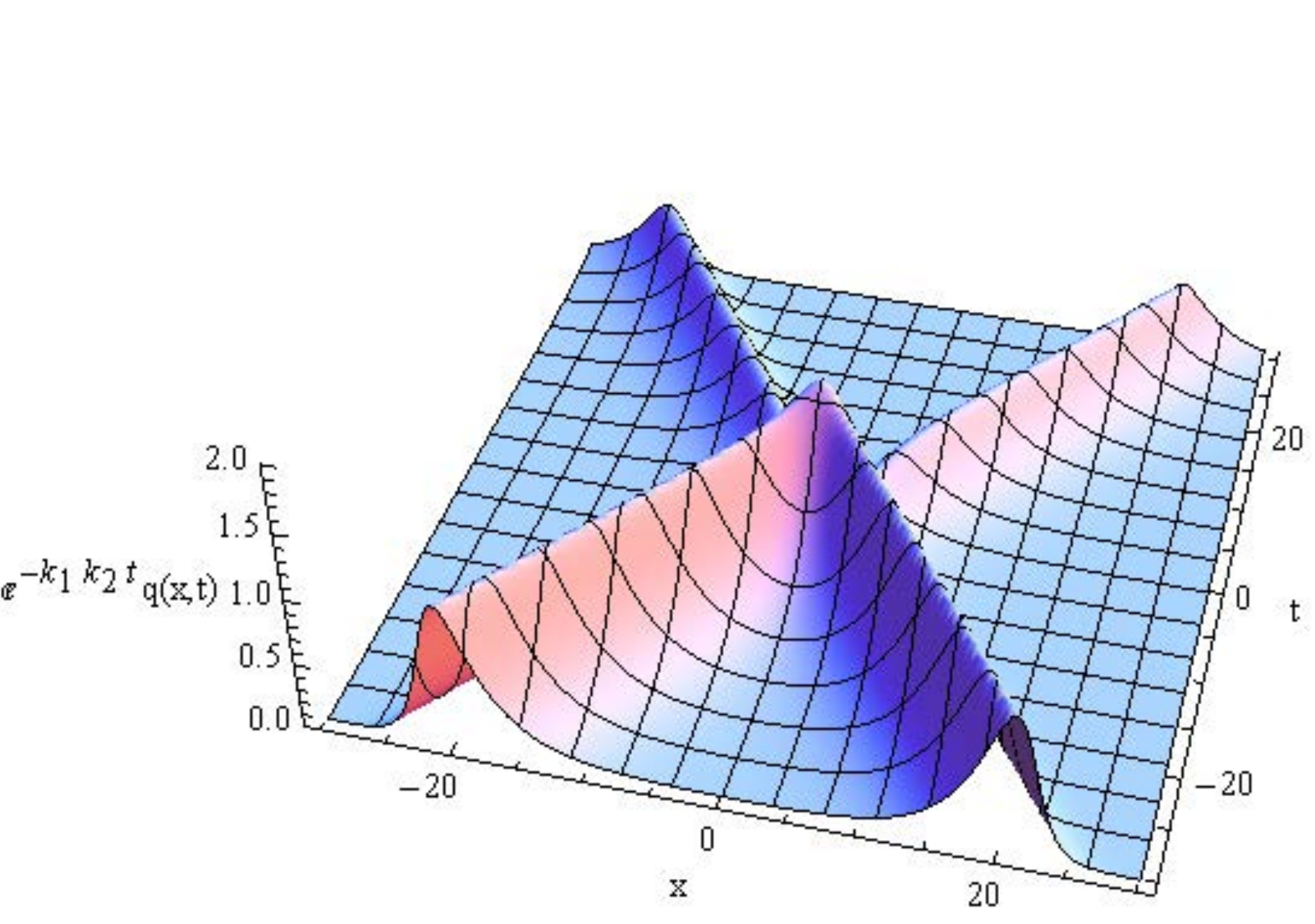}
}
\\
\caption{(a). Shape and motion of 2SS \eqref{nls2-2ss-2} for equation \eqref{nls-2-eq-2},
in which $k_1=-0.3, k_2=1.0$.
(b). Shape and motion of $e^{-k_1k_2 t}q(x,t)$ where $q(x,t)$ is depicted in (a).
\label{Fig-3}}
\end{figure}
%==========================================

We make asymptotic analysis for $q(x,t)$ with a compensation factor $e^{-k_1k_2t}$ for the case of $0<-k_1<k_2$.
Consider $e^{-k_1k_2t}q(x,t)$ respectively on lines
\begin{equation}
x+(k_2+k_1)t=y,~~ x-(k_2+k_1)t=z,
\label{yz-2}
\end{equation}
where $y$ and $z$ can be viewed as new variables.
Then we rewrite $e^{-k_1k_2t}q(x,t)$ in coordinate frame $\{y,t\}$,
fix $y$, let $t\to \pm\infty$ and we find
\begin{equation*}
e^{-k_1k_2t}q(x,t)\sim \left\{
        \begin{array}{ll}
         \frac{2k_1(k_1^2-k_2^2)}{(k_1+k_2)^2 e^{k_1 y}-4k_1k_2  e^{k_2 y}},& t\to \infty,\\
         \frac{2k_2(k_1^2-k_2^2) }{4k_1k_2  e^{k_1 y}-(k_2+k_1)^2 e^{k_2 y}},& t\to -\infty.
        \end{array}
        \right.
\end{equation*}
It indicates that
along the line $y= \mathrm{constant}$, there is a soliton;
from $t=-\infty$ to $t=\infty$, the soliton gains a shift
\begin{equation}
\Delta=\frac{2}{k_1-k_2} \ln \Bigl(\frac{k_2+k_1}{4k_1k_2}\Bigr)\label{shift-2}
\end{equation}
and its amplitude changes from
\[A_{-}^{(y)}=-(k_1+k_2)\frac{k_2}{k_1}\cdot 2^{\frac{k_1+k_2}{k_1-k_2}} \cdot \Bigl(\frac{k_1+k_2}{k_1}\Bigr)^{\frac{2k_1}{k_2-k_1}}\]
to
\[A_{+}^{(y)}=-(k_1+k_2)\frac{k_1}{k_2}\cdot 2^{\frac{k_1+k_2}{k_2-k_1}} \cdot \Bigl(\frac{k_1+k_2}{k_2}\Bigr)^{\frac{2k_2}{k_1-k_2}};\]
the soliton itself is not symmetric in terms of shape.
If we rewrite $-e^{k_1k_2t}q(x,t)$ in coordinate frame $\{z,t\}$,
fix $z$ and let $t\to \pm\infty$, we have
\begin{equation*}
e^{-k_1k_2t}q(x,t)\sim \left\{
        \begin{array}{ll}
         \frac{2k_1(k_1^2-k_2^2)}{(k_1+k_2)^2 e^{-k_1 z}-4k_1k_2 e^{-k_2 z}},& t\to \infty,\\
         \frac{2k_2(k_1^2-k_2^2)}{4k_1k_2 e^{-k_1 z}-(k_1+k_2)^2 e^{-k_2 z}},& t\to -\infty.
        \end{array}
        \right.
\end{equation*}
Thus, asymptotically,
along the line $z= \mathrm{constant}$, there is a soliton;
from $t=-\infty$ to $t=\infty$, the soliton gains a shift \eqref{shift-2}
and its amplitude changes from $A_{-}^{(z)}=A_{-}^{(y)}$ to $A_{+}^{(z)}=A_{+}^{(y)}$;
the soliton itself is not symmetric in terms of shape.
Again, the trajectories of two solitons are completely symmetric  in $\{x,t\}$ plane.

\subsection{Solutions of the NLS(I), mKdV and cmKdV hierarchies}\label{sec-3-3}

For the NLS(I), mKdV and cmKdV hierarchies \eqref{nls-hie1}, \eqref{mkdv-hie} and \eqref{cmkdv-hie},
we skip proofs which are analogues to the NLS(II) case, and
we present constraint relations and solutions of these hierarchies.
In the following we always suppose $T\in \mathbb{C}_{2(n+1)\times 2(n+1)}$.

\begin{theorem}\label{T:NLS-I}
For the NLS(I) hierarchy \eqref{nls-hie1}, consider hierarchy \eqref{akns-hie-even} and $\varphi$ and $\psi$ defined in \eqref{akns-var-psi-A-c-d-even}.
When we impose constraint $m=n$,
\begin{equation}
\psi(x,t)=T \varphi^*(\sigma x,t),~~ C^-=TC^{+*},
\label{constr-NLS-I}
\end{equation}
and require $T$ to meet
\begin{equation}
AT+\sigma TA^*=0 ,~~ TT^*=\delta \sigma I,
\label{NLS-I-TA}
\end{equation}
then the NLS(I) hierarchy \eqref{nls-hie1} admit solutions
\begin{equation}\label{nls-I-q-solu}
 q(x)= \frac{2|\widehat{\varphi}^{(n-1)}; \widehat{\psi}^{(n+1)}|}{|\widehat{\varphi}^{(n)}; \widehat{\psi}^{(n)}|}.
\end{equation}
If  $T$ and $A$ are block matrices \eqref{nls2-example-tac} where
$T_i$ and $A_i$ are $(n+1)\times (n+1)$ matrices,
then, solutions to \eqref{NLS-I-TA} are given as below.
\begin{table}[H]
\begin{center}
\begin{tabular}{|c|c|c|}
\hline
   $(\sigma, \delta)$    & $T$ &  $A$    \\
\hline
   $(1,-1)$              & $T_1=T_4=\mathbf{0}_{n+1}$, $T_3=-T_2 =\mathbf{I}_{n+1}$ & $ K_1=-K^*_4=\mathbf{K}_{n+1}$ \\
\hline
   $(1,1)$              & $T_1=T_4=\mathbf{0}_{n+1}$, $T_3=T_2 =\mathbf{I}_{n+1}$ & $ K_1=-K^*_4=\mathbf{K}_{n+1}$ \\
\hline
   $(-1,-1)$              & $T_1=T_4=\mathbf{0}_{n+1}$, $T_3=T_2 =\mathbf{I}_{n+1}$ & $ K_1=K^*_4=\mathbf{K}_{n+1}$ \\
\hline
   $(-1,1)$              & $T_1=T_4=\mathbf{0}_{n+1}$, $T_3=-T_2 =\mathbf{I}_{n+1}$ & $ K_1=K^*_4=\mathbf{K}_{n+1}$ \\
\hline
\end{tabular}
\caption{$T$ and $A$ for the NLS(I) hierarchy}
\label{Tab-2}
\end{center}
\end{table}

In \eqref{nls-I-q-solu} the vector $\varphi$ can be chosen as the following.
When $\mathbf{K}_{n+1}$ is the diagonal matrix \eqref{K-n},
we have %($c_1^-$ or $d_1$)
\begin{equation}\label{nls-II-varphi-psi}
\varphi= (c_1 e^{\theta(k_1)}, c_2 e^{\theta(k_2)},\cdots,c_{n+1} e^{\theta(k_{n+1})},
 d_1 e^{\theta(-k_1^*)}, d_2 e^{\theta(-k_2^*)},\cdots,d_{n+1} e^{\theta(-k_{n+1}^*)})^T
\end{equation}
%\end{subequations}
for the case of $\sigma=1$, and
\begin{equation}\label{nls-II-varphi-psi-1}
\varphi= (c_1 e^{\theta(k_1)}, c_2 e^{\theta(k_2)},\cdots,c_{n+1} e^{\theta(k_{n+1})},
 d_1 e^{\theta(k_1^*)}, d_2 e^{\theta(k_2^*)},\cdots,d_{n+1} e^{\theta(k_{n+1}^*)})^T
\end{equation}
%\end{subequations}
for the case of $\sigma=-1$, where
\begin{equation}\label{nls-I-scatter}
\theta(k_l)=\frac{1}{2}k_l x+\frac{i}{2}\sum^{\infty}_{j=1}k_l^{2j}t_{2j}.
\end{equation}
When  $\mathbf{K}_{n+1}$ is a $(n+1)\times (n+1)$ Jordan matrix w.r.t. $J_{n+1}(k)$ defined as in \eqref{Jordan},
we have
 \begin{equation}\label{nls-I-varphi-jord}
 \varphi= \left(c e^{\theta(k)}, \frac{\partial_{k}}{1!}(c e^{\theta(k)}),\cdots, \frac{\partial^n_{k}}{n!} (c e^{\theta(k)}),
 d e^{\theta(-k^*)}, \frac{\partial_{k^*}}{1!}(d e^{\theta(-k^*)}),\cdots,\frac{\partial^n_{k^*}}{n!}(d e^{\theta(-k^*)})\right)^T
 \end{equation}
for the case of $\sigma=1$, and
 \begin{equation}\label{nls-I-varphi-jord}
 \varphi= \left(c e^{\theta(k)}, \frac{\partial_{k}}{1!}(c e^{\theta(k)}),\cdots, \frac{\partial^n_{k}}{n!} (c e^{\theta(k)}),
 d e^{\theta(k^*)}, \frac{\partial_{k^*}}{1!}(d e^{\theta(k^*)}),\cdots,\frac{\partial^n_{k^*}}{n!}(d e^{\theta(k^*)})\right)^T
 \end{equation}
for the case of $\sigma=-1$.

\end{theorem}

As examples we list one-soliton solutions for the representative equation \eqref{nls-eq1} of the NLS(I) type:
\begin{subequations}
\begin{align}
& q_{\sigma=1,\delta=-1}=-\frac{(k+k^*)c^* d^*}{|c|^2 e^{kx+ik^2 t_2}+ |d|^2e^{-k^*x + ik^{*2}  t_2}},\label{1ss-h}\\
& q_{\sigma=1,\delta=1}=\frac{(k+k^*)c^* d^*}{|c|^2e^{kx + ik^{2} t_2}-|d|^2 e^{-k^*x + ik^{*2} t_2}},\\
& q_{\sigma=-1,\delta=-1}=\frac{(k-k^*)c^* d^*}{|c|^2 e^{kx+ik^2 t_2}-|d|^2 e^{k^* x+ ik^{*2}  t_2}},\\
& q_{\sigma=-1,\delta=1}=-\frac{(k - k^*)c^* d^*}{|c|^2e^{kx+ik^2t_2}+|d|^2 e^{k^* x+ ik^{*2}  t_2}},\label{1ss-ndf}
\end{align}
\end{subequations}
where we take $k=k_1$ and $C^+=(c,d)^T$.

\vskip 10pt

\begin{theorem}\label{T:mkdv}
For the mKdV hierarchy \eqref{mkdv-hie}, consider hierarchy \eqref{akns-hie-odd} and $\varphi$ and $\psi$ defined in \eqref{akns-var-psi-A-c-d-odd}.
In double Wronskians we take $m=n$ and impose constraint
\begin{equation}
\psi(x,t)=T \varphi(\sigma_1 x, \sigma_2 t),~~ C^-=TC^{+},
\label{constr-mkdv}
\end{equation}
where $T$ satisfies
\begin{equation}
AT+\sigma_1 TA=0 ,~ TT=\delta \sigma_1 I.
\label{mkdv-TA}
\end{equation}
The mKdV hierarchy \eqref{mkdv-hie} admit solutions
\begin{equation}\label{mkdv-q-solu}
 q(x)= \frac{2|\widehat{\varphi}^{(n-1)}; \widehat{\psi}^{(n+1)}|}{|\widehat{\varphi}^{(n)}; \widehat{\psi}^{(n)}|}.
\end{equation}
If  $T$ and $A$ are block matrices \eqref{nls2-example-tac} where
$T_i$ and $A_i$ are $(n+1)\times (n+1)$ matrices,
then, solutions to \eqref{mkdv-TA} are given as below. %\textcolor[rgb]{1.00,0.00,0.00}{(check)}
\begin{table}[H]
\begin{center}
\begin{tabular}{|c|c|c|}
\hline
   $(\sigma_i, \delta)$    & $T$ &  $A$    \\
\hline
   $(1,1)$              & $T_1=T_4=\mathbf{0}_{n+1}$, $T_3=T_2 =\mathbf{I}_{n+1}$ & $ K_1=-K_4=\mathbf{K}_{n+1}$ \\
\hline
   $(1,-1)$              & $T_1=T_4=\mathbf{0}_{n+1}$, $T_3=-T_2 =\mathbf{I}_{n+1}$ & $ K_1=-K_4=\mathbf{K}_{n+1}$ \\
\hline
   $(-1,1)$              & $T_1=-T_4=i\mathbf{I}_{n+1}$, $T_3=T_2 =\mathbf{0}_{n+1}$ & $ K_1=-K_4=\mathbf{K}_{n+1}$ \\
\hline
   $(-1,-1)$              & $T_1=-T_4=\mathbf{I}_{n+1}$, $T_3=T_2 =\mathbf{0}_{n+1}$ & $ K_1=-K_4=\mathbf{K}_{n+1}$ \\
\hline
\end{tabular}
\caption{$T$ and $A$ for the mKdV hierarchy}
\label{Tab-3}
\end{center}
\end{table}

In solution \eqref{mkdv-q-solu} the vector $\varphi$ can be chosen as the following.
When $\mathbf{K}_{n+1}$ is the diagonal matrix \eqref{K-n},
we have
\begin{equation}\label{mkdv-varphi-psi}
\varphi= (c_1 e^{\xi(k_1)}, c_2 e^{\xi(k_2)},\cdots,c_{n+1} e^{\xi(k_{n+1})},
 d_1 e^{\xi(-k_1)}, d_2 e^{\xi(-k_2)},\cdots,d_{n+1} e^{\xi(-k_{n+1})})^T,
\end{equation}
where $\xi(k_i)$ is defined by \eqref{kdv-scatter}.
When  $\mathbf{K}_{n+1}$ is a $(n+1)\times (n+1)$ Jordan matrix w.r.t. $J_{n+1}(k)$ defined as in \eqref{Jordan},
we have
 \begin{equation}\label{mkdv-varphi-jord}
 \varphi= \left(c e^{\xi(k)}, \frac{\partial_{k}}{1!}(c e^{\xi(k)}),\cdots, \frac{\partial^n_{k}}{n!} (c e^{\xi(k)}),
 d e^{\xi(-k)}, \frac{\partial_{k}}{1!}(d e^{\xi(-k)}),\cdots,\frac{\partial^n_{k}}{n!}(d e^{\xi(-k)})\right)^T.
 \end{equation}

\end{theorem}

\vskip 10pt
One-soliton solutions for the representative equation \eqref{mkdv-eq} of the mKdV type are respectively
\begin{align*}
& q_{\sigma_i=1,\delta=-1}=-\frac{2k c d}{c^2 e^{kx+ k^{3}t_3} + d^2 e^{-kx-k^3t_3}},~~ (i=1,2),\\
& q_{\sigma_i=-1,\delta=-1}=\frac{2k}{e^{kx+k^3t_3}+e^{-kx-k^{3}t_3}}, ~~ (i=1,2),
\end{align*}
where we take $k=k_1$ and $C^+=(c,d)^T$. Here we only write out solutions for two equations in \eqref{mkdv-eq},
due to the following remark.

\begin{remark}\label{rem-3}
On the level of equations,
\eqref{mkdv-hie} with $(\sigma_i, \delta)=(\pm1,1)$ and \eqref{mkdv-hie} with $(\sigma_i, \delta)=(\pm1,-1)$ can be transformed from each other
by simply taking $q\to iq$.
Such a transformation can also be seen on the level of solutions we obtained.
First, for $q_{(\sigma_i, \delta)=(1,-1)}$, due to arbitrariness  of vector $C^{+}$,
if we replace $C^{+}$ by $\mathrm{Diag}(\mathbf{I}_{n+1}, i \mathbf{I}_{n+1})C^{+}$,
then, $q_{(\sigma_i, \delta)=(1,-1)}$ defined by \eqref{mkdv-q-solu}  will be changed to $iq_{(\sigma_i, \delta)=(1,-1)}$,
which gives a solution to \eqref{mkdv-hie} with $(\sigma_i, \delta)=(1,1)$.
Second, for $q$ given in the form \eqref{mkdv-q-solu}, from the line $(\sigma_i, \delta)=(-1,-1)$ and line $(\sigma_i, \delta)=(-1,1)$ in Table \ref{Tab-3},
one can easily find the relation $q_{(\sigma_i, \delta)=(-1,-1)}=iq_{(\sigma_i, \delta)=(-1,1)}$.
With these connections we only need to consider the cases $(\sigma_i, \delta)=(1,-1)$ and $(\sigma_i, \delta)=(-1,-1)$.
\end{remark}

\begin{remark}\label{rem-4}
Due to the special block structure of the matrix $T$ for the cases of $(\sigma_i,\delta)=(-1,\pm 1)$ in Table \ref{Tab-3},
both $C^{\pm}$ can be gauged to be $(1,1,\cdots, 1)^T$,
which means the solutions we obtain for the cases $(\sigma_i,\delta)=(-1,\pm 1)$
are independent of phase parameters $C^{\pm}$.
\end{remark}

\begin{theorem}\label{T:cmkdv}
For the cmKdV hierarchy \eqref{cmkdv-hie}, consider hierarchy \eqref{akns-hie-odd} and $\varphi$ and $\psi$ defined in \eqref{akns-var-psi-A-c-d-odd}.
In double Wronskians we take $m=n$ and impose constraint
\begin{equation}
\psi(x,t)=T \varphi^*(\sigma_1 x,\sigma_2 t),~~ C^-=TC^{+*},
\label{constr-cmkdv}
\end{equation}
where $T$ satisfies
\begin{equation}
AT+\sigma_1 TA^*=0 ,~ TT^*=\delta \sigma_1 I.
\label{cmkdv-TA}
\end{equation}
The cmKdV hierarchy \eqref{cmkdv-hie} admit solutions
\begin{equation}\label{cmkdv-q-solu}
 q(x,t)= \frac{2|\widehat{\varphi}^{(n-1)}; \widehat{\psi}^{(n+1)}|}{|\widehat{\varphi}^{(n)}; \widehat{\psi}^{(n)}|}.
\end{equation}
If  $T$ and $A$ are block matrices \eqref{nls2-example-tac} where
$T_i$ and $A_i$ are $(n+1)\times (n+1)$ matrices,
then, solutions to \eqref{cmkdv-TA} are given as below, %\textcolor[rgb]{1.00,0.00,0.00}{(check)}
\begin{table}[H]
\begin{center}
\begin{tabular}{|c|c|c|}
\hline
   $(\sigma_i, \delta)$    & $T$ &  $A$    \\
\hline
   $(1,1)$              & $T_1=T_4=\mathbf{0}_{n+1}$, $T_3=T_2 =\mathbf{I}_{n+1}$ & $ K_1=-K^*_4=\mathbf{K}_{n+1}$ \\
\hline
   $(1,-1)$              & $T_1=T_4=\mathbf{0}_{n+1}$, $T_3=-T_2 =\mathbf{I}_{n+1}$ & $ K_1=-K^*_4=\mathbf{K}_{n+1}$ \\
\hline
   $(-1,1)$              & $T_1=T_4=\mathbf{0}_{n+1}$, $T_3=-T_2 =\mathbf{I}_{n+1}$ & $ K_1=K^*_4=\mathbf{K}_{n+1}$ \\
\hline
   $(-1,-1)$              & $T_1=T_4=\mathbf{0}_{n+1}$, $T_3=T_2 =\mathbf{I}_{n+1}$ & $ K_1=K^*_4=\mathbf{K}_{n+1}$ \\
\hline
\end{tabular}
\caption{$T$ and $A$ for the cmKdV hierarchy}
\label{Tab-4}
\end{center}
\end{table}

In solution \eqref{cmkdv-q-solu} the vector $\varphi$ can be chosen as the following.
When $\mathbf{K}_{n+1}$ is the diagonal matrix \eqref{K-n},
we have
\begin{equation}\label{cmkdv-varphi-psi}
\varphi= (c_1 e^{\xi(k_1)}, c_2 e^{\xi(k_2)},\cdots,c_{n+1} e^{\xi(k_{n+1})},
 d_1 e^{\xi(-k_1^*)}, d_2 e^{\xi(-k_2^*)},\cdots,d_{n+1} e^{\xi(-k_{n+1}^*)})^T
\end{equation}
for the case of $\sigma_i=1$, and
\begin{equation}\label{cmkdv-varphi-psi-1}
\varphi= (c_1 e^{\xi(k_1)}, c_2 e^{\xi(k_2)},\cdots,c_{n+1} e^{\xi(k_{n+1})},
 d_1 e^{\xi(k_1^*)}, d_2 e^{\xi(k_2^*)},\cdots,d_{n+1} e^{\xi(k_{n+1}^*)})^T
\end{equation}
for the case of $\sigma_i=-1$  where $\xi(k_i)$ is defined by \eqref{kdv-scatter}.
When  $\mathbf{K}_{n+1}$ is a $(n+1)\times (n+1)$ Jordan matrix w.r.t. $J_{n+1}(k)$ defined as in \eqref{Jordan},
we have
 \begin{equation}\label{cmkdv-varphi-jord}
 \varphi= \left(c e^{\xi(k)}, \frac{\partial_{k}}{1!}(c e^{\xi(k)}),\cdots, \frac{\partial^n_{k}}{n!} (c e^{\xi(k)}),
 d e^{\xi(-k^*)}, \frac{\partial_{k^*}}{1!}(d e^{\xi(-k^*)}),\cdots,\frac{\partial^n_{k^*}}{n!}(d e^{\xi(-k^*)})\right)^T
 \end{equation}
 for the case of $\sigma_i=1$, and
 \begin{equation}\label{cmkdv-varphi-jord-1}
 \varphi= \left(c e^{\xi(k)}, \frac{\partial_{k}}{1!}(c e^{\xi(k)}),\cdots, \frac{\partial^n_{k}}{n!} (c e^{\xi(k)}),
 d e^{\xi(k^*)}, \frac{\partial_{k^*}}{1!}(d e^{\xi(k^*)}),\cdots,\frac{\partial^n_{k^*}}{n!}(d e^{\xi(k^*)})\right)^T
 \end{equation}
 for the case of $\sigma_i=-1$.

\end{theorem}

\vskip 10pt
One-soliton solutions for the representative equation \eqref{cmkdv-eq} of the cmKdV type are respectively
\begin{align*}
& q_{\sigma_i=1,\delta=1}=\frac{(k+k^*)c^* d^*}{|c|^2 e^{kx + k^{3} t_3} - |d|^2 e^{-k^* x - k^{*3} t_3}},\\
& q_{\sigma_i=1,\delta=-1}= - \frac{(k+k^*)c^* d^*}{|c|^2 e^{kx + k^{3} t_3} + |d|^2 e^{-k^* x - k^{*3} t_3}},\\
& q_{\sigma_i=-1,\delta=1}=-\frac{(k-k^*)c^* d^*}{|c|^2 e^{kx + k^{3}t_3} + |d|^2 e^{k^*x + k^{*3} t_3}},\\
& q_{\sigma_i=-1,\delta=-1}=\frac{(k-k^*)c^* d^*}{|c|^2 e^{kx + k^{3}t_3} - |d|^2 e^{k^*x + k^{*3} t_3}},
\end{align*}
where we take $k=k_1$ and $C^+=(c,d)^T$.

\section{Solutions of local and nonlocal sine-Gordon}\label{sec-4}

%%%%%%%%%%%%%%%%%%%%%%%%%%%%%%%%%%%%%%%%%%%%%%%%%%%%%%%%%%%%%%%%%%
In this section, we consider the first member of the negative order ANKS hierarchy (AKNS$(-1)$ for short),
its reductions to the sine-Gordon (sG) and nonlocal sG equation.

\subsection{AKNS$(-1)$: reductions}\label{sec-4-1}

Let us first go through the AKNS$(-1)$ system and its reductions.
The AKNS$(-1)$ system reads(cf.\cite{ZhaJZ-PD-2009})
\begin{equation}\label{akns-1-equation}
 q_{xt}-2q\partial^{-1}_x(qr)_t=q,~~r_{xt}-2r\partial^{-1}(qr)_t=r,
\end{equation}
where the integration operator $\partial^{-1}$ is defined as
\begin{equation}
\partial^{-1}_x\, \cdot =\frac{1}{2}(\int^{x}_{-\infty}-\int^{\infty}_{x})\,\cdot\, \mathrm{d}x.
\label{int-op}
\end{equation}
Since both $q$ and $r$ tend to zero as $|x|\to \infty$ and $qr$ is a conserved density,
one can alternatively write \eqref{akns-1-equation} as
\begin{equation}\label{akns-1-equation-sqr}
q_{xt}=2qs,~r_{xt}=2rs,~s_x=(qr)_t,
\end{equation}
where $s(x,t)$ is an auxiliary function satisfying
\begin{equation}
 s(x,t)=\partial^{-1}_x (qr)_t  + s_0,~~ (s_0= s(x,t)|_{|x|\to \infty}=-\frac{1}{2}).
\label{s-qr}
\end{equation}
Here we note that \eqref{akns-1-equation-sqr} is known as
the generalized coupled integrable dispersionless system \cite{konno-1996}.

One known reduction of \eqref{akns-1-equation-sqr}  is due to $r=-q$ and the result is \cite{konno-1994}
\begin{equation}\label{akns-1-equation-sq}
q_{xt}-2qs=0,~~s_x+2q q_t=0.
\end{equation}
This system connects with the sG equation $u_{xt}=\sin u$ through two ways.
One way is that via\cite{HirT-JPSJ-1994} $q=\frac{u_x}{2},~ s=\frac{1}{2}\cos u$ the system yields
\begin{align*}
& q_{xt}-2qs=\frac{1}{2}(u_{xt}-\sin u)_x=0,\\
& s_x+2q q_t=- \frac{1}{2} u_x (u_{xt}-\sin u)=0.
\end{align*}
In this sense we call \eqref{akns-1-equation-sq} a non-potention form of the sG equation.
The other way to sG equation $u_{xy}=\sin u$ is via a hodograph transformation (cf. \cite{KonO-JPSJ-1994}).
A second known reduction  of \eqref{akns-1-equation-sqr} is due to $r=-q^*$, giving a complex system \cite{Kon-AA-1995}
\begin{equation}\label{akns-1-equation-sqc}
q_{xt}-2qs=0,~~s_x+ (|q|^2)_t=0.
\end{equation}
Bilinear forms and dynamics of  solutions of  the AKNS$(-1)$ system \eqref{akns-1-equation} and its reduced
\eqref{akns-1-equation-sq} and \eqref{akns-1-equation-sqc} have been investigated in \cite{ZhaJZ-PD-2009}.
Recently, geometric formulation of \eqref{akns-1-equation-sqc} and its hodograph link with a complex short pulse equation
were given in \cite{FenMO-SAPM-2016}.

System \eqref{akns-1-equation-sqr} admits a nonlocal reduction \cite{ablowitz-study-2016}
\begin{equation}\label{sine-gordon-non-reduction}
 r(x,t)= - q(-x,-t),
\end{equation}
where $q,r\in \mathbb{R}[x,t]$.
Under the definition \eqref{s-qr} with the integration operator \eqref{int-op} and the above reduction \eqref{sine-gordon-non-reduction}, we find
\begin{align*}
s(-x,-t)=& \frac{1}{2}(\int^{-x}_{-\infty}-\int^{\infty}_{-x})(q(x',-t)q(-x',t))_t \mathrm{d}x' -\frac{1}{2}  \\
        =& -\frac{1}{2}(\int^{x}_{\infty}-\int^{-\infty}_{x})(q(-y,-t)q(y,t))_t \mathrm{d}y -\frac{1}{2}  \\
        =& -\frac{1}{2}(\int^{x}_{-\infty}-\int^{\infty}_{x})(q(-y,-t)q(y,t))_t \mathrm{d}y -\frac{1}{2},
\end{align*}
which yields a relation\footnote{Note that $\partial_x^{-1}=\int^{x}_{-\infty} \mathrm{d}x$ used in \cite{ablowitz-study-2016}
is not sufficient to get relation \eqref{ss1}.}
\begin{equation}
s(-x,-t)=s(x,t).
\label{ss1}
\end{equation}
Thus, with definition \eqref{s-qr}, system \eqref{akns-1-equation-sqr} yields
\begin{equation}\label{sine-gordon-nonlocal}
q_{xt}(x,t)-2q(x,t)s(x,t)=0,~~s_x(x,t)+[q(x,t)q(-x,-t)]_{t}=0,
\end{equation}
which is called real nonlocal sG equation.
In fact, if $q(x,t)$ is a solution to \eqref{sine-gordon-nonlocal}, so is $q(-x,-t)$, due to relation \eqref{ss1}.
A complex nonlocal reduction can be taken as
\begin{equation}\label{sine-gordon-non-reduction}
 r(x,t)= -q^*(-x,-t),
\end{equation}
where $q,r\in \mathbb{C}[x,t]$, which yields relation $s(-x,-t)=s^*(x,t)$ and a complex nonlocal sG equation
\begin{equation}\label{sine-gordon-nonlocal-c}
q_{xt}(x,t)-2q(x,t)s(x,t)=0,~~s_x(x,t)+[q(x,t)q^*(-x,-t)]_{t}=0.
\end{equation}

\subsection{AKNS$(-1)$: double Wronskian solutions}\label{sec-4-2}

Employing the rational transformation \eqref{akns-transformation},
AKNS$(-1)$ \eqref{akns-1-equation} can be written into a bilinear system\cite{ZhaJZ-PD-2009}
\begin{equation}
 D_x D_t h\cdot f= h f,~~D_x D_t g\cdot f= g f,~~D_{x}^2 f\cdot f= 2hg,
\end{equation}
which has double Wronskian solutions\cite{song2010}
\begin{equation}\label{akns--1-fgh}
 f=|\widehat{\varphi}^{(n)};\widehat{\psi}^{(m)}|,~g=  2|\widehat{\varphi}^{(n+1)};\widehat{\psi}^{(m-1)} |,
 ~h= 2 |\widehat{\varphi}^{(n-1)}; \widehat{\psi}^{(m+1)} |,
\end{equation}
where $\varphi$ and $\psi$ are  defined by
\begin{equation}\label{phipsi-anks--1}
 \varphi = \exp{ \Bigl( -Ax-\frac{1}{4}A^{-1}t \Bigr)}C^+,~~\psi = \exp{ \Bigl( Ax+\frac{1}{4}A^{-1}t \Bigr)}C^-,
\end{equation}
$A\in \mathbb{C}_{(n+m+2)\times(n+m+2)}$ and $C^{\pm}$ are $(n+m+2)$-th order onstant column vectors.

\subsection{Solutions to the reduced AKNS$(-1)$}

Solutions for the reduced AKNS$(-1)$ system can be discussed via a procedure similar to the NLS and mKdV cases.
In the following we skip details and directly list out results.

\begin{theorem}\label{T:sG-r}
For the AKNS$(-1)$ system \eqref{akns-1-equation} (or \eqref{akns-1-equation-sqr}), its real reduction
\begin{equation}
r(x,t)=-q(\sigma x,\sigma t),~~\sigma=\pm 1
\label{red-sg-r}
\end{equation}
yields the real nonpotential sG equation \eqref{akns-1-equation-sq} and its nonlocal version  \eqref{sine-gordon-nonlocal}.
Corresponding reductions on the double Wronskians \eqref{akns--1-fgh} with \eqref{phipsi-anks--1} can be implemented by
taking $m=n$ and
\begin{equation}
\psi(x,t)=T \varphi(\sigma x,\sigma t),~~ C^-=TC^{+},
\label{constr-sg-r}
\end{equation}
where $T$ and $A$ satisfy
\begin{equation}
AT+\sigma TA=0 ,~ T^2= -\sigma I.
\label{sg-r-TA}
\end{equation}
Under the above constraints, solutions to the real nonpotential sG equation \eqref{akns-1-equation-sq} and its nonlocal version  \eqref{sine-gordon-nonlocal}
are given as the form
\begin{equation}\label{sg-r-q-solu}
 q(x,t)= \frac{2|\widehat{\varphi}^{(n-1)}; \widehat{\psi}^{(n+1)}|}{|\widehat{\varphi}^{(n)}; \widehat{\psi}^{(n)}|}.
\end{equation}
When  $T$ and $A$ are block matrices \eqref{nls2-example-tac} where
$T_i$ and $A_i$ are $(n+1)\times (n+1)$ matrices,
 solutions to \eqref{sg-r-TA} are given as below, %\textcolor[rgb]{1.00,0.00,0.00}{(check)}
\begin{table}[H]
\begin{center}
\begin{tabular}{|c|c|c|}
\hline
   $\sigma$    & $T$ &  $A$    \\
\hline
   $1$              & $T_1=T_4=\mathbf{0}_{n+1}$, $T_3=-T_2 =\mathbf{I}_{n+1}$ & $ K_1=-K_4=\mathbf{K}_{n+1}$ \\
\hline
   $-1$              & $T_1=-T_4=-\mathbf{I}_{n+1}$, $T_3=T_2 =\mathbf{0}_{n+1}$ & $ K_1=-K_4=\mathbf{K}_{n+1}$ \\
\hline
\end{tabular}
\caption{$T$ and $A$ for the real nonpotential sG system}
\label{Tab-5}
\end{center}
\end{table}

In solution \eqref{sg-r-q-solu} the vector $\varphi$ can be chosen as the following.
When $\mathbf{K}_{n+1}$ is the diagonal matrix \eqref{K-n},
we have
\begin{equation}\label{sg-r-varphi-psi}
\varphi= (c_1 e^{\mu(k_1)}, c_2 e^{\mu(k_2)},\cdots,c_{n+1} e^{\mu(k_{n+1})},
 d_1 e^{\mu(-k_1)}, d_2 e^{\mu(-k_2)},\cdots,d_{n+1} e^{\mu(-k_{n+1})})^T,
\end{equation}
where
\begin{equation}\label{mu}
\mu(k_i)=-k_i x -\frac{t}{4k_i},~~ k_i\in \mathbb{R}.
\end{equation}
When  $\mathbf{K}_{n+1}$ is a $(n+1)\times (n+1)$ Jordan matrix w.r.t. $J_{n+1}(k)$ defined as in \eqref{Jordan},
we have
 \begin{equation}\label{sg-r-varphi-jord}
 \varphi= \left(c e^{\mu(k)}, \frac{\partial_{k}}{1!}(c e^{\mu(k)}),\cdots, \frac{\partial^n_{k}}{n!} (c e^{\mu(k)}),
 d e^{\mu(-k)}, \frac{\partial_{k}}{1!}(d e^{\mu(-k)}),\cdots,\frac{\partial^n_{k}}{n!}(d e^{\mu(-k)})\right)^T.
 \end{equation}
Note that for $\sigma=-1$ case, $C^{\pm}$ can be always gauged to $(1,1\cdots,1)^T$ without changing solution $q(x,t)$.

\end{theorem}

\begin{theorem}\label{T:sG-c}
For the AKNS$(-1)$ system \eqref{akns-1-equation} (or \eqref{akns-1-equation-sqr}), its complex reduction
\begin{equation}
r(x,t)=-q^*(\sigma x,\sigma t),~~\sigma=\pm 1
\label{red-sg-r}
\end{equation}
yields complex nonpotential sG equation \eqref{akns-1-equation-sqc} and its nonlocal version  \eqref{sine-gordon-nonlocal-c}.
Corresponding reductions on the double Wronskians \eqref{akns--1-fgh} with \eqref{phipsi-anks--1} can be implemented by
taking $m=n$ and
\begin{equation}
\psi(x,t)=T \varphi^*(\sigma x,\sigma t),~~ C^-=TC^{+*},
\label{constr-sg-r}
\end{equation}
where $T$ and $A$ satisfy
\begin{equation}
AT+\sigma TA^*=0 ,~ TT^*=-\sigma I.
\label{sg-r-TA-com}
\end{equation}
Under the above constraints, solutions to  \eqref{akns-1-equation-sqc} and its nonlocal version  \eqref{sine-gordon-nonlocal-c}
are given as the form  \eqref{sg-r-q-solu}.
When  $T$ and $A$ are block matrices \eqref{nls2-example-tac} where
$T_i$ and $A_i$ are $(n+1)\times (n+1)$ matrices,
 solutions to \eqref{sg-r-TA-com} are given as below, %\textcolor[rgb]{1.00,0.00,0.00}{(check)}
\begin{table}[H]
\begin{center}
\begin{tabular}{|c|c|c|}
\hline
   $\sigma$    & $T$ &  $A$    \\
\hline
   $1$              & $T_1=T_4=\mathbf{0}_{n+1}$, $T_3=-T_2 =\mathbf{I}_{n+1}$ & $ K_1=-K_4^*=\mathbf{K}_{n+1}$ \\
\hline
   $-1$              & $T_1=T_4=\mathbf{0}_{n+1}$, $T_3=T_2 =-\mathbf{I}_{n+1}$ & $ K_1=K_4^*=\mathbf{K}_{n+1}$ \\
\hline
\end{tabular}
\caption{$T$ and $A$ for the complex nonpotential sG system}
\label{Tab-5}
\end{center}
\end{table}

In solution \eqref{sg-r-q-solu} the vector $\varphi$ can be chosen as following. When $\mathbf{K}_{n+1}$ is the diagonal matrix \eqref{K-n}, we have
\begin{equation}
\varphi= (c_1 e^{\mu(k_1)}, c_2 e^{\mu(k_2)},\cdots,c_{n+1} e^{\mu(k_{n+1})},
 d_1 e^{\mu(-k_1^*)}, d_2 e^{\mu(-k_2^*)},\cdots,d_{n+1} e^{\mu(-k_{n+1}^*)})^T
\end{equation}
for the case of $\sigma=1$, and
\begin{equation}
\varphi= (c_1 e^{\mu(k_1)}, c_2 e^{\mu(k_2)},\cdots,c_{n+1} e^{\mu(k_{n+1})},
 d_1 e^{\mu(k_1^*)}, d_2 e^{\mu(k_2^*)},\cdots,d_{n+1} e^{\mu(k_{n+1}^*)})^T
\end{equation}
for the case of $\sigma=-1$,
where $\mu(k_i)$ is defined by \eqref{mu}. When $\mathbf{K}_{n+1}$ is a $(n+1)\times (n+1)$ Jordan matrix w.r.t. $J_{n+1}(k)$ defined as in \eqref{Jordan},
we have
 \begin{equation*}
 \varphi= \left(c e^{\mu(k)}, \frac{\partial_{k}}{1!}(c e^{\mu(k)}),\cdots, \frac{\partial^n_{k}}{n!} (c e^{\mu(k)}),
 d e^{\mu(-k^*)}, \frac{\partial_{k^*}}{1!}(d e^{\mu(-k^*)}),\cdots,\frac{\partial^n_{k^*}}{n!}(d e^{\mu(-k^*)})\right)^T
 \end{equation*}
 for the case of $\sigma=1$, and
  \begin{equation*}
 \varphi= \left(c e^{\mu(k)}, \frac{\partial_{k}}{1!}(c e^{\mu(k)}),\cdots, \frac{\partial^n_{k}}{n!} (c e^{\mu(k)}),
 d e^{\mu(k^*)}, \frac{\partial_{k^*}}{1!}(d e^{\mu(k^*)}),\cdots,\frac{\partial^n_{k^*}}{n!}(d e^{\mu(k^*)})\right)^T
 \end{equation*}
 for the case of $\sigma=-1$.

\end{theorem}

\vskip 10pt
One-soliton solutions for the systems \eqref{akns-1-equation-sq}, \eqref{akns-1-equation-sqc},
\eqref{sine-gordon-nonlocal} and \eqref{sine-gordon-nonlocal-c} are respectively
\begin{align}
& q_{\eqref{akns-1-equation-sq}}=\frac{4k c d}{c^2 e^{-2k x - \frac{1}{2k}t} + d^2 e^{2k x + \frac{1}{2k}t}},\label{q-89}\\
& q_{\eqref{akns-1-equation-sqc}}= \frac{2(k+k^*) c^* d^* }{|c|^2 e^{-2k x - \frac{1}{2k}t}   +  |d|^2 e^{2k^* x + \frac{1}{2k^*}t}}, \\
& q_{\eqref{sine-gordon-nonlocal}}=\frac{4k }{e^{-2k x - \frac{1}{2k}t} + e^{2k x + \frac{1}{2k}t}},\label{nxtsg}\\
& q_{\eqref{sine-gordon-nonlocal-c}}=\frac{2(k-k^*) c^* d^* }{|c|^2 e^{-2k x - \frac{1}{2k}t}  -|d|^2 e^{-2k^* x - \frac{1}{2k^*}t}},\label{q95}
\end{align}
where we take $k=k_1$ and $C^+=(c,d)^T$.

As an example we consider dynamics of nonlocal sG equation \eqref{sine-gordon-nonlocal}. Its 1SS given in \eqref{nxtsg}
describes a single soliton
\begin{equation*}
q(x,t)=2k \,\mathrm{sech} \Bigl(2k x + \frac{1}{2k}t\Bigr),
\end{equation*}
travelling with a fixed amplitude $2k$ and speed $\frac{1}{4k^2}$.
2SS of equation \eqref{sine-gordon-nonlocal} is written as
\begin{equation}
q(x,t)=\frac{G}{F},
\label{2ss-sg}
\end{equation}
where
\begin{align*}
G= & 8\, e^{\frac{(k_1 + k_2) (t + 4 k_1 k_2 x)}{ 2 k_1 k_2}} (k_1^2 - k_2^2) \Bigl[-k_2 \cosh\Bigl(\frac{t}{2 k_1} + 2 k_1 x\Bigr) +
    k_1 \cosh\Bigl(\frac{t}{2k_2} + 2 k_2 x\Bigr)\Bigr],\\
F=& (1 + e^{\frac{t}{k_1} + 4 k_1 x}) (1 + e^{\frac{t}{k_2} + 4 k_2 x})( k_1^2 + k_2^2)\\
  & +  2 k_1 k_2\Bigl(-1 +e^{\frac{t}{k_1} + 4 k_1 x}+  e^{\frac{t}{k_2} + 4 k_2 x}-
     4  e^{\frac{(k_1 + k_2) (t + 4 k_1 k_2 x)}{ 2 k_1 k_2}} -  e^{\frac{(k_1 + k_2) (t + 4 k_1 k_2 x)}{  k_1 k_2}}\Bigr).
\end{align*}
Solution corresponding to $2\times 2$ Jordan matrix case is
\begin{equation}
q(x,t)=\frac{8 k^2 [2 k \cosh(2 k x+\frac{t}{2 k}) + (t - 4 k^2 x) \sinh( 2 k x +\frac{t}{2 k})]}
{t^2 + 16 k^4 x^2 + k^2 (2 - 8 t x) + 2 k^2 \cosh(4 k x +\frac{t}{k})}.
\label{sg-jss}
\end{equation}

%============================  Fig ===================================
\begin{figure}[!h]
\centering \subfigure[]
{\label{fig-4-1} %% label for first subfigure
\includegraphics[width=2.4in]{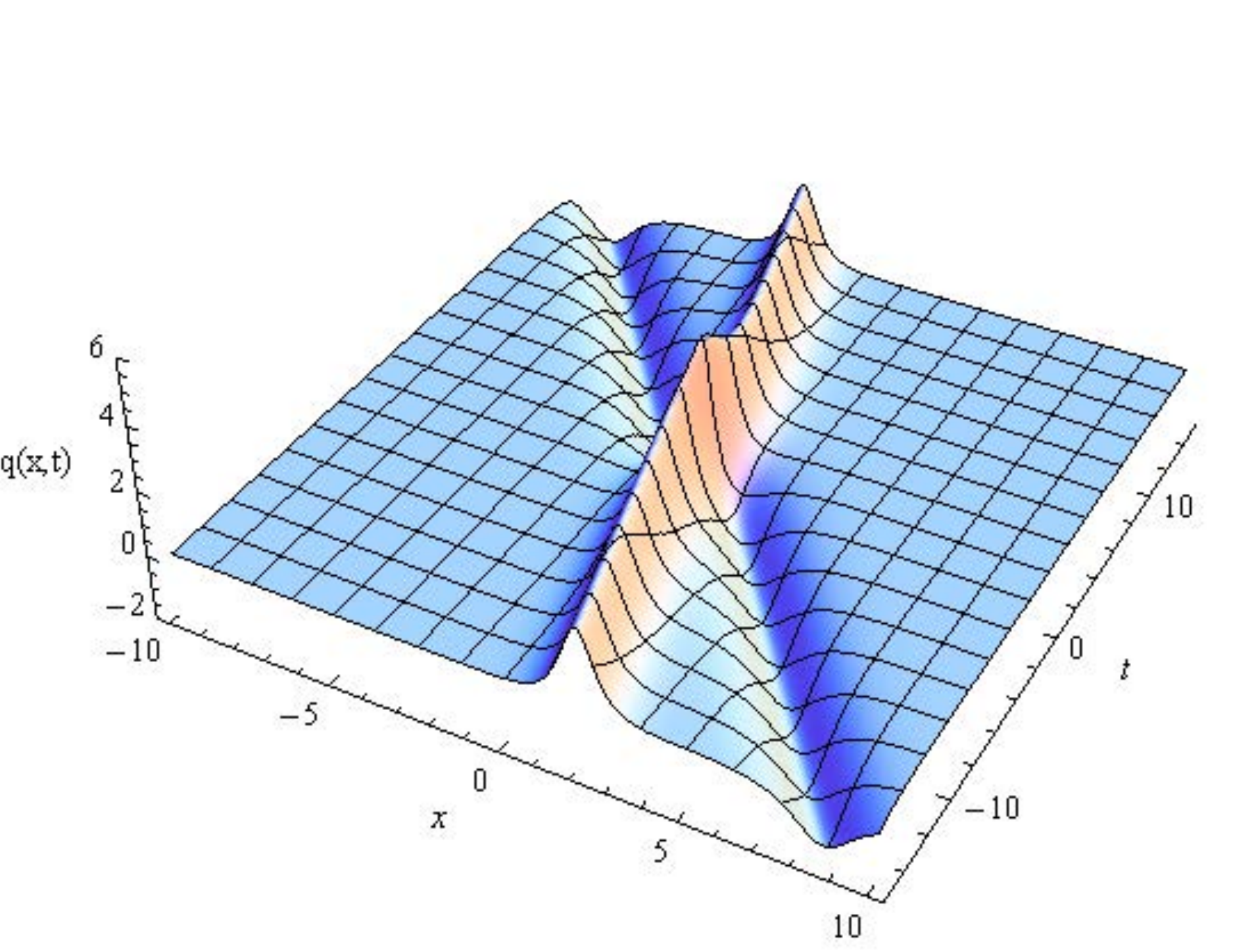}}~~~
\subfigure[]{
\label{fig-4-2} %% label for second subfigure
\includegraphics[width=2.4in]{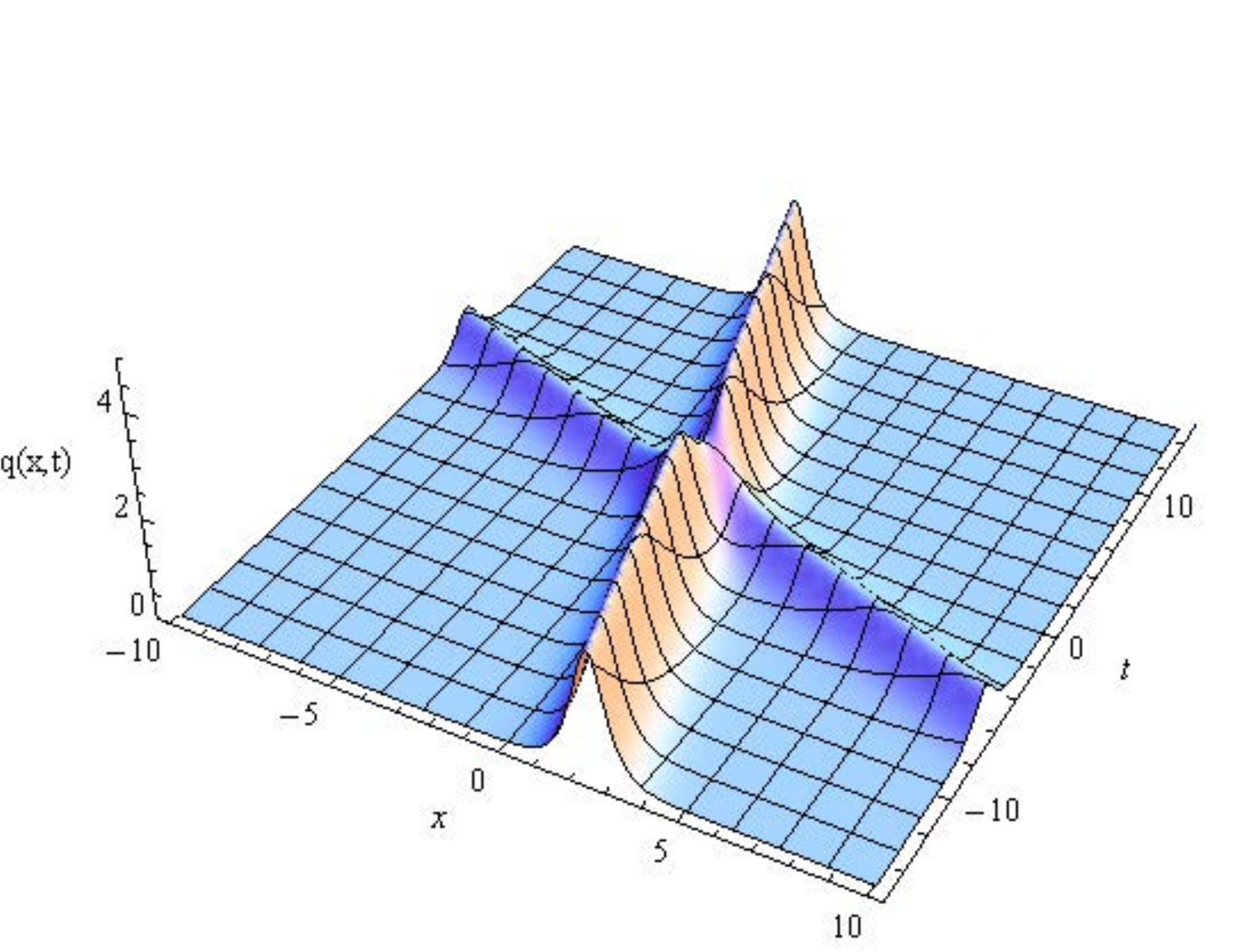}
}\\
\centering \subfigure[]
{\label{fig-4-3} %% label for first subfigure
\includegraphics[width=2.4in]{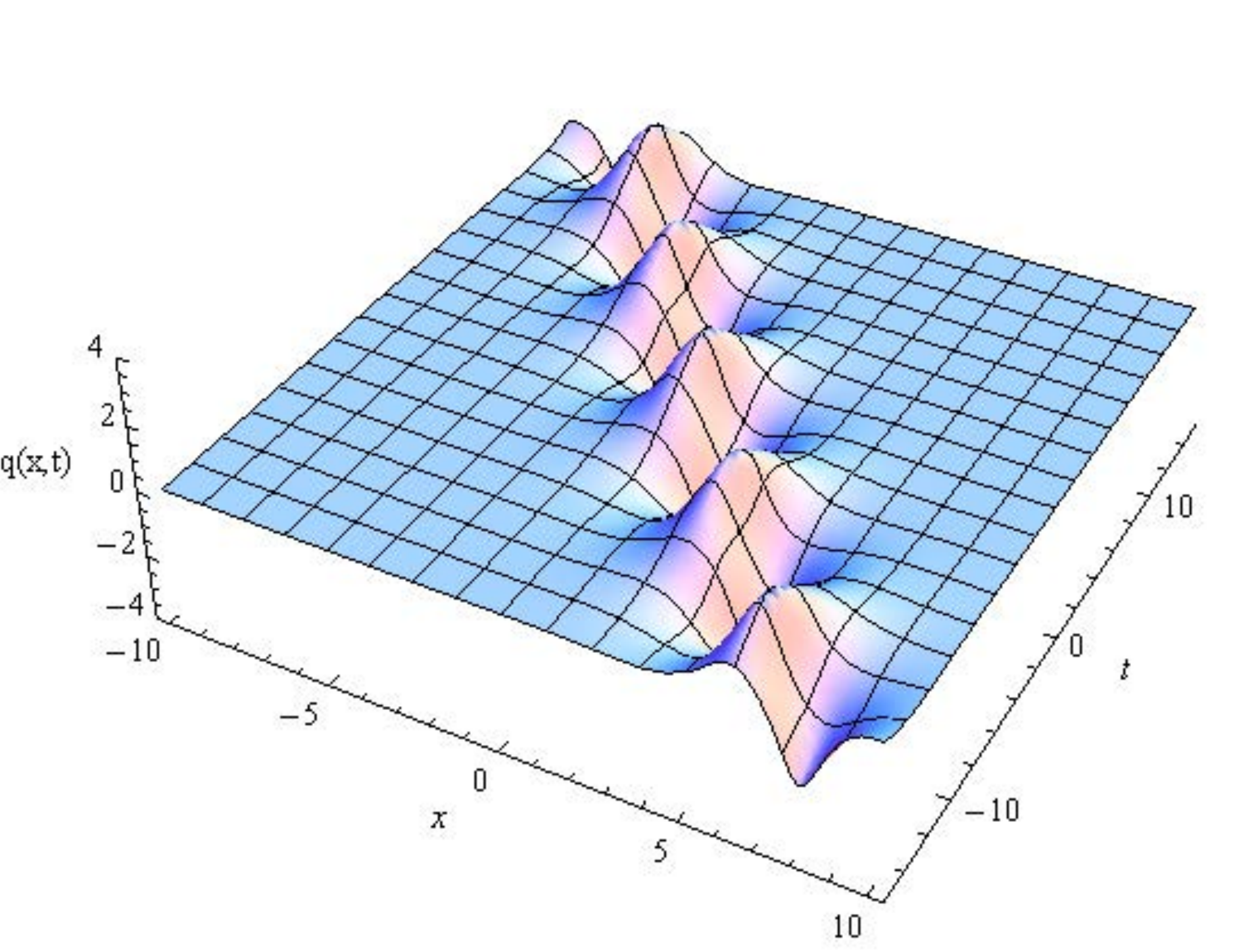}}~~~
\subfigure[]{
\label{fig-4-4} %% label for second subfigure
\includegraphics[width=2.4in]{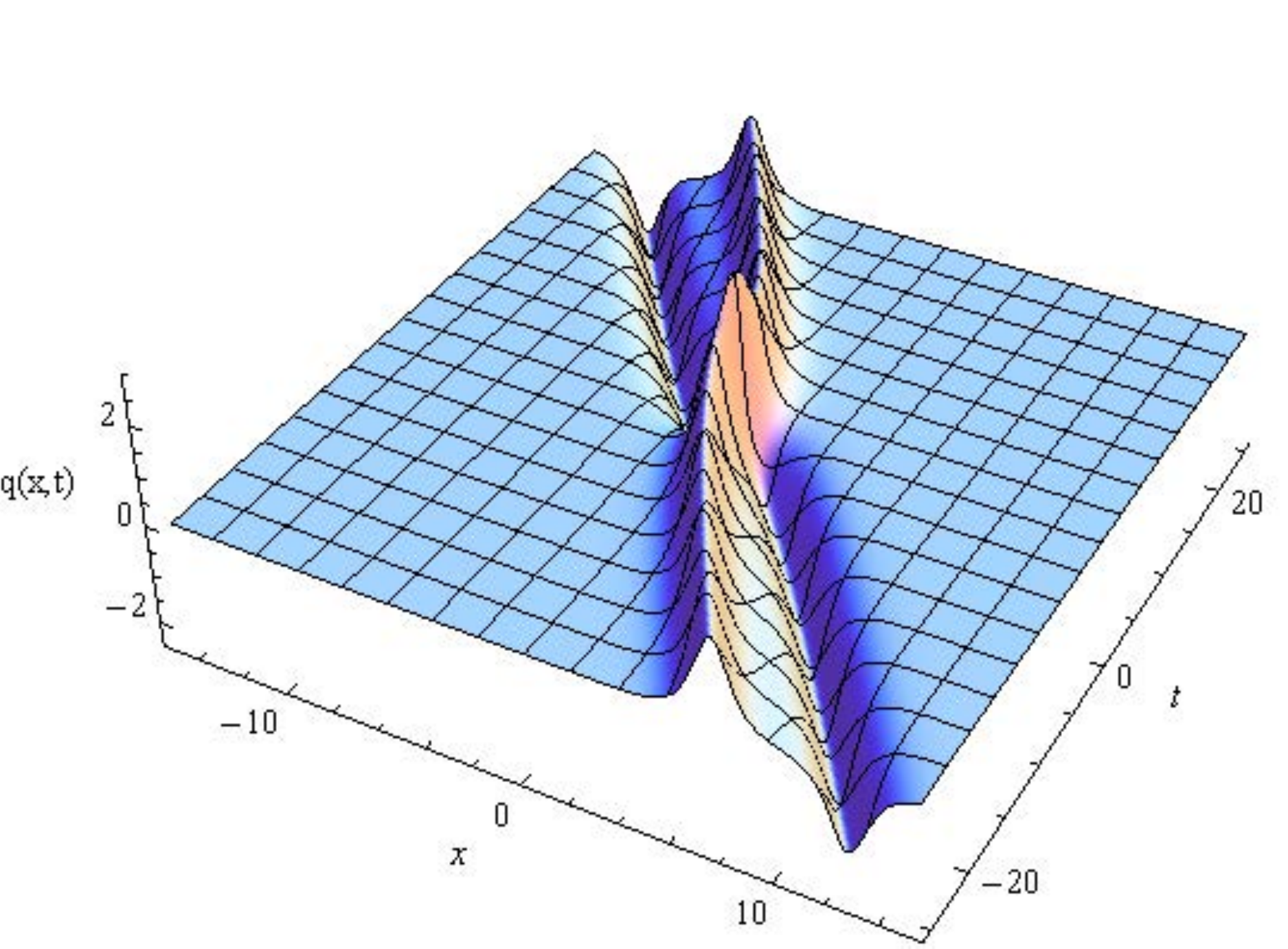}
}
\\
\caption{(a). Shape and motion of 2SS \eqref{2ss-sg} for equation \eqref{sine-gordon-nonlocal},
in which $k_1=0.7, k_2=1.2$.
(b). Shape and motion of 2SS \eqref{2ss-sg} for equation \eqref{sine-gordon-nonlocal}, in which $k_1=-0.4, k_2=1.2$.
(c). Breather provided by  \eqref{2ss-sg} for equation \eqref{sine-gordon-nonlocal}, in which $k_1=0.5+0.5i, k_2=0.5-0.5i$.
(d). Shape and motion of \eqref{sg-jss} for equation \eqref{sine-gordon-nonlocal}, in which $k=0.8$.
\label{Fig-4}}
\end{figure}
%==========================================

These solutions are illustrated in Fig.\ref{Fig-4}. Fig.\ref{fig-4-1} and Fig.\ref{fig-4-2} show normal interaction of two solitons,
Fig.\ref{fig-4-3} describes a breather which is generated when taking $k_1=k_2^*=a+ib$, where $a,b\in \mathbb{R}$ and $ab\neq 0$.
In Fig.\ref{Fig-4}(d) two solitons are completely in symmetric shape and their trajectories governed by logarithm  function with a linear background.
These dynamics are quite similar to the mKdV$^{+}$ equation (cf.\cite{ZhaZSZ-RMP-2014})
and asymptotic analysis can be implemented with standard prodecure.
We skip details which one may refer to \cite{ZhaZSZ-RMP-2014}
in which the mKdV$^{+}$ equation was analyzed.

\section{Conclusions an discussions}\label{sec-5}

In this paper we have shown a reduction technique that enables us to obtain solutions for
the reduced local and nonlocal equations from the double Wronskian solutions of the original
before-reduction systems.
In particular, we obtained solutions for the KdV hierarchy,
mKdV hierarchy, NLS hierarchy, sG equation (in nonpotential form)
as well as some nonlocal equations, such as the NLS hierarchy with reverse space,
with reverse time, with reverse space-time, and so forth.
As an example we illustrated two-soliton interactions of the reverse-time NLS equation \eqref{nls-2-eq} and \eqref{nls-2-eq-2}.
Its single soliton is always stationary, but the trajectories of two solitons are completely symmetric in $\{x,t\}$ plane
and for Eq.\eqref{nls-2-eq} their amplitudes are also related to phase parameters.
This is different from usual soliton systems.  Asymptotic analysis has been given as demonstration.
We also illustrated solutions of reverse-$(x,t)$ sG equation. They exhibit dynamics similar to the usual
mKdV$^{+}$ equation.

For each reduced equation,
the reduction provides a coefficient matrix $A$ (see \eqref{akns-var-psi-A-c-d} as an example)
with a constraint structure.
Since the eigenvalues of $A$ correspond to discrete spectrum of the AKNS
spectral problem  \eqref{akns-spectral},
the structures of $A$ for each reduced equation
indicate how the distribution of discrete spectrum changes with different reductions.
For example, Table \ref{Tab-2} indicates that, before reduction,
the  discrete spectrum data are
\begin{equation}
\{k_1, k_2, \cdots, k_{n+1}; k_{n+2}, k_{n+3}, \cdots, k_{2n+2}\},~~ k_j\in \mathbb{C},
\label{sp-dis}
\end{equation}
while reduction $r(x,t)=-q^*(x,t)$ provides constraints which force \eqref{sp-dis} to be
\[\{k_1, k_2, \cdots, k_{n+1}\} \cup \{ k_{n+1+j}= k^*_{j}, ~j=1,2,\cdots, n+1\},~~ k_j\in \mathbb{C},\]
and due to $r(x,t)=q^*(-x,t)$, \eqref{sp-dis} has to be
\[\{k_1, k_2, \cdots, k_{n+1}\} \cup \{ k_{n+1+j}=- k^*_{j}, ~j=1,2,\cdots, n+1\},~~ k_j\in \mathbb{C}.\]
The above correspondence, together with the fact that
line $(\sigma, \delta)=(1,-1)$ and $(\sigma, \delta)=(-1,1)$ share  same $T$, explains why 
by formally replacing $\{ k^*_{j}\}$ with  $\{-k^*_{j}\}$  the soliton solutions of the focusing  NLS equation
\begin{equation}
iq_t=q_{xx}+2 |q|^2q
\label{nls}
\end{equation}
yield solutions to the reverse space defocusing  NLS equation
\begin{equation}
iq_t(x,t)=q_{xx}(x,t)-2 q^2(x,t)q^*(-x,t).
\label{nls-d}
\end{equation}
Same trick holds for some other equations, e.g. the cmKdV case. 

The above $k^*_{j} \to -k^*_{j}$ trick also provides a realization for the formal transformations 
between local and nonlocal equations given in \cite{YanY-arxiv-2017}.
For example, the focusing NLS equation
\begin{equation}
i\hat{q}_{\hat t}(\hat x, \hat t)=\hat q_{\hat x \hat x}+2 |\hat q|^2\hat q
\label{nls-hat}
\end{equation}
and nonlocal defocusing NLS equation \eqref{nls-d} are connected by\footnote{Here we add a multiplier $i^s $ where $s\in \mathbb{Z}$,
which does not change equation \eqref{nls-hat}.}
\begin{equation}
x= i \hat x,~~ t=-\hat t,~~ q(x,t)=i^s \hat q(\hat x, \hat t).
\label{tr-yy}
\end{equation}
Note that it is required that after the transformation both $x$ and $\hat x$ must be considered to be real in 
the equations \eqref{nls-d} and \eqref{nls-hat},
which looks contradictory to the relation $x= i \hat x$ in \eqref{tr-yy}.
However, it is possible to switch real $x$ and real $\hat x$  on a suitable platform of solutions.
Let us take one-soloton solution \eqref{1ss-h} as an example to explain the mechanism.
With notation $(\hat x, \hat t)$  \eqref{1ss-h} reads
\begin{equation}
\hat q(\hat x, \hat t)=-\frac{(k+k^*)c^* d^*}{|c|^2 e^{k\hat x+ik^2 \hat t}+ |d|^2e^{-k^*\hat x + ik^{*2}  \hat t}}.
\label{1ss-hh}
\end{equation}
Since $k$ is arbitrary, we  replace $k$ with $ik$ and get
\begin{equation*}
\hat q(\hat x, \hat t)=-\frac{i(k-k^*)c^* d^*}{|c|^2 e^{i k\hat x - ik^2 \hat t}+ |d|^2e^{i k^*\hat x - ik^{*2}  \hat t}},
%\label{1ss-hhh}
\end{equation*} 
which is still a solution of \eqref{nls-hat}.
Then we implement the replacement \eqref{tr-yy} with $s=-1$ and we reach
\[q(x,t)=-\frac{(k - k^*)c^* d^*}{|c|^2e^{kx+ik^2t}+|d|^2 e^{k^* x+ ik^{*2}  t}},
\] 
which is nothing but \eqref{1ss-ndf}, the one-soliton solution to the nonlocal defocusing NLS equation \eqref{nls-d}.
In practice, for multisoliton soltions in double Wronskian form one can always first replace $\{k_j\}$ by $\{ik_j\}$ in $\hat q(\hat x, \hat t)$, then 
make replacement \eqref{tr-yy}, and the resulting  $q(x,t)$ is a multisoliton soltion of Eq.\eqref{nls-d}.
Since after replacing $\{k_j\}$ by $\{ik_j\}$ one should keep $i\hat x$ unchanged 
in conjugation operation (as in this turn $i\hat x$ is to be replaced by real $x$),  
this procedure is equivalent to directly play  $k^*_{j} \to -k^*_{j}$ trick in the double Wronskian solution of Eq.\eqref{nls}.
With the above mechanism, all the transformations between local equations and their nonlocal counterparts 
given in \cite{YanY-arxiv-2017} are meaningful and practical.

To summarize, we have given a reduction approach to get explicit solutions of nonlocal equations from those of before-reduction system,
while \cite{YanY-arxiv-2017} gave formal transformations between nonlocal equations and local counterparts. 
There is a third approach developed in \cite{Cau-SAPM-2017} where the nonlocal NLS equation (e.g. \eqref{nls-d}) is considered as
a combination of two equations (on half line $x\geq 0$) that are derived from a matrix version of the AKNS spectral problem \eqref{akns-spectral}
in which $q$ and $r$ are replaced by $Q=\left(\begin{smallmatrix} q(x,t) & 0 \\0&  q(-x,t) \end{smallmatrix}\right)$
and  $R=\left(\begin{smallmatrix} r(x,t) & 0 \\0&  r(-x,t) \end{smallmatrix}\right)$,
imposed a local reduction $R=\mu Q^* \mu^{-1}$ where $\mu=\left(\begin{smallmatrix} 0 & 1 \\ 1 &  0 \end{smallmatrix}\right)$.
As a consequence, one can derive  multisolitons of the nonlocal case
from the multisolitons of the unreduced matrix AKNS hierarchy, applying local reductions.
All these approaches allow us to present solutions  of nonlocal equations by making use of known results.

\subsection*{Acknowledgments}
This work was supported by  the NSF of China [grant numbers 11371241,  11435005, 11631007].

\end{document}